\documentclass{article} 
\usepackage[english]{babel}
\usepackage[top=2cm, bottom=1.8cm,left=2.5cm,right=2.5cm,head=2cm,headsep=0.5cm]{geometry} 

\usepackage{lmodern}

\usepackage{bbm}
\usepackage{amsmath}		
\usepackage{amssymb, dsfont}		
\usepackage{mathrsfs}
\usepackage{latexsym}	
\usepackage{array}
\usepackage{multirow}
\usepackage{amsthm, array, braket, color}
\usepackage{graphicx}
\usepackage{hyperref}
\usepackage{verbatim}
\usepackage{authblk}
\usepackage{tikz}
\usetikzlibrary{matrix}

\usepackage[toc,page]{appendix} 
\usepackage{enumitem}
\usepackage{csquotes}

\renewcommand{\epsilon}{\varepsilon}

\newcommand{\E}{\mathbb{E}}

\newcommand{\R}{\mathbb{R}}
\newcommand{\Z}{\mathbb{Z}}
\newcommand{\C}{\mathbb{C}}
\newcommand{\N}{\mathbb{N}}
\renewcommand{\P}{\mathbb{P}}

\newcommand{\Un}{\mathds{1}}

\newcommand{\norm}[1]{{\left\|{#1}\right\|}}
\newcommand{\ent}[1]{{\left[{#1}\right]}}
\newcommand{\abs}[1]{{\left|{#1}\right|}}
\newcommand{\scal}[1]{{\left\langle{#1}\right\rangle}}

\newcommand{\cala}{\mathcal{A}}

\newcommand{\cc}{\mathcal{C}}
\newcommand{\dd}{\mathcal{D}}

\newcommand{\ff}{\mathcal{F}}
\newcommand{\calg}{\mathcal{G}}
\newcommand{\pp}{\mathcal{P}}
\newcommand{\hh}{\mathcal{H}}

\newcommand{\mm}{\mathcal{M}}

\newcommand{\call}{\mathcal{L}}
\newcommand{\bb}{\mathcal{B}}

\newcommand{\cals}{\mathcal{S}}
\newcommand{\Tau}{\mathcal{T}}
\newcommand{\vv}{\mathcal{V}}
\newcommand{\calw}{\mathcal{W}}
\newcommand{\xx}{\mathcal{X}}
\newcommand{\yy}{\mathcal{Y}}

\newcommand{\rrho}{\varrho}
\newcommand{\eps}{\epsilon}
\newcommand{\ssig}{\varsigma}

\newcommand{\gu}{\mathfrak{U}}
\newcommand{\gs}{\mathfrak{S}}

\newcommand{\tr}[1]{\text{Tr}\left( {#1} \right) }
\newcommand{\tra}{\text{Tr} }

\newtheorem{theo}{Theorem}
\newtheorem{prop}[theo]{Proposition}
\newtheorem{defi}[theo]{Definition}
\newtheorem{lem}[theo]{Lemma}
\newtheorem{cor}[theo]{Corollary}

\theoremstyle{definition}
\newtheorem{rem}{Remark}

\newtheorem{assumption}{Assumption}

\begin{document}
\title{The Open Quantum Brownian Motion and continual measurements}
\author{Simon Andr\'eys}
\maketitle
\tableofcontents

\begin{abstract}
This article is a mathematical analysis of the Open Quantum Brownian Motion. This object was introduced in \cite{OQBM} as the limit of a family of Open Quantum Random Walks on the graph $\Z$. We prove the convergence for the three possible descriptions of this object: the quantum trajectory satisfying a Belavkin Equation, the unitary evolution on the Fock space satisfying a quantum Langevin Equation, and the Lindbladian evolution. We introduce a very general framework for the continual measurement of non-demolition observables, which is applied to the measurement of the position of the Open Quantum Brownian Motion, and we probe some questions related to the convergence of processes in this context.
\end{abstract}

\section{Introduction}

\subsection{General introduction}

 Open Quantum Random Walks (OQW) are a quantum generalization of discrete Markov chains and were introduced by Attal, Petruccione, Sabot and Sinayskyi in  \cite{OQWbirth}.
 They consists into a particle moving randomly on a discrete graph with transition probability depending on its internal quantum state. They model a quantum system subject to dissipation or repeated measurement with control, and are used for example as a toy model to study coherence in photosynthetic cells \cite{decoherence_assisted_transport}; they have been the subject of extensive mathematical study, see the end of Paragraph \ref{subsub:oqw} for more references.
 While OQW are defined on discrete graph and on discrete space, the Open Quantum Brownian Motion (OQBM) was introduced in \cite{OQBM} to model a particle moving on the line in continuous time. It was defined as the limit of a family of OQW on $\Z$, with a diffusive normalization, i.e. with a time scale $\tau$ going to zero and a space scale ${\delta=\sqrt{\tau}}$.
 The obtained process depends in two operators $N$ and $H$; in the trivial case where ${N=H=0}$ the classical Brownian motion is recovered. The Open Quantum Brownian Motion has been derived from a microscopic physical model in \cite{SinayskiyPetruccione_micro_derivation} and \cite{SinayskiyPetruccione_control_OQBM}.
 A mathematically interesting phenomenon was observed on the OQBM, namely the transition from diffusive to ballistic behavior as the parameters $N$ and $H$ are changed \cite{BauerBernardTilloy_ballistic} with the appearance of so-called spikes in the ballistic regime \cite{spikes1} \cite{spikes2}, which were then studied in the context of more general stochastic differential equations \cite{spikes3}, \cite{kolb_spikes_2019} and \cite{bernardin_spiking_2018}.
 
As for OQWs, the OQBM has three different descriptions. It can be seen has a Lindblad evolution ${\rho_t=\Lambda_S^t(\rho_0)}$ on the Hilbert space ${\hh_G\otimes L^2(\R)}$, where $\hh_G$ represents the internal state of the particle. The second description is a Stinespring dilation ${\rho_{tot, t}=\gu_t (\rho_0\otimes\ket{\Omega}\bra{\Omega}) \gu_t^*}$ on ${\hh_G\otimes L^2(\R)\otimes \Phi}$, where ${\Phi}$ is the Fock space and $\gu_t$ satisfies a Quantum Stochastic Differential Equation (QSDE) called the Hudson-Parthasaraty Equation. This representation is more complete than the Lindbladian one, since it allows to compute the quantum correlation between the events at two different times. Finally, upon the continual measure of the position of the particle, it admits a quantum trajectories unraveling, that is a random process $(\rrho_t, X_t)_{t\in \R}$ where $\rrho_t$ is a random state on $\hh_G$ and $X_t\in \R$ is a random position. When $\hh_G$ is of finite dimension this process obeys a classical stochastic differential equation, often called the diffusive Belavkin Equation \cite{BoutenGutaMaasen} \cite{Belavkin92}.

In the original article on the OQBM \cite{OQBM}, most results where derived formally but not rigorously proved. The main purpose of this article is to explicit the mathematical meaning of the statements of \cite{OQBM}, pointing out some of the mathematical issues and completing the proofs.

In the second part of the introduction, we introduce OQWs and the formal definition of the OQBM, and the mathematical problems raised by this definition, which are tackled in the rest of the article. Besides the problem of the convergence, a mathematical issue appears in the description of the Lindbladian: for an OQW, the evolution projects the states on the set of \emph{diagonal state}, i.e. states of the form $\rho=\sum_{x\in \vv} \rho(x)\otimes\ket{x}\bra{x}\in \gs(\hh_G\otimes L^2(\vv))$, where $\vv$ is the set of vertices of the graph on which the particle is moving. In the continuous case, diagonal operators are replaced by multiplication operators of the form $\int_{\R} \rho(x) d\ket{x}\bra{x}$, which cannot be trace class and hence cannot be a state. Hence, the discrete object which converges to the continuous OQBM is actually not an OQW in the strict meaning of the term, though it coincides with an OQW on the set of diagonal states. 

In the second section, we introduce the repeated measurement model and the quantum stochastic calculus and we prove the convergence of the discrete models for the OQBM to the continuous one in each description: for the unraveled process, we prove a convergence in distribution in the Skorokhod space as a direct consequence of a theorem of Pellegrini \cite{PellegriniDiffusive08}. For the unitary dilations, the strong convergence of the unitary operators is proved from a theorem of Attal and Pautrat \cite{AttalPautrat}. This strong convergence allows to prove the strong convergence for the Lindblad operators. 

In the third section we look into another claim of the article \cite{OQBM}, in which the unraveled process $(\rrho_t, X_t)_{t\in [0, T]}$ is obtained from the continual measurement of an observable under the evolution by the unitary operators $\gu_t$. This makes use of the quantum filtering theory \cite{GoughNotes} \cite{BoutenGutaMaasen} \cite{Belavkin92} and the notion of non-demolition measurement. We introduce rigorously the continual measurement of non-demolition observables in a way which is equivalent to the quantum filtering approach but we believe is more adapted to the Schr\"odinger picture of the evolution, and we apply it to the case of the OQBM. Finally, we ponder the relation between the convergence of the unitary operators $\gu_t$ and the convergence in distribution of the unraveling, obtaining only an incomplete result which generates a few open questions. 
\vspace{0.5cm}

{\bf Acknowledgment:} I thank Clement Pellegrini and Yann Pautrat for useful remarks and advice, Denis Bernard and Antoine Tilloy for their talks on the Open Quantum Brownian Motion, Ivan Bardet who pointed out some of the mathematical issues raised by their article, and my advisor St\'ephan Attal for encouraging me to write on this subject and suggested numerous improvements to the presentation of the paper.

\subsection{Open Quantum Random Walks and the Open Quantum Brownian Motion}
In this subsection we introduce the notion of Open and Unitary Quantum Walks (OQW and UQW) and we describe the formal definitions of the Open Quantum Brownian Motion (OQBM) and the related mathematical issues.

\subsubsection{General notations}

The basic object in quantum mechanics is a separable Hilbert space $\hh$ (all Hilbert spaces are implicitly supposed to be separable in this article). Let us gather some of the notations and definitions we will use:
\begin{itemize}
\item The identity operator on $\hh$ (respectively $\hh_A$ and $\C^n$) is written $I_\hh$ (respectively $I_A$ and $I_n$) or simply $I$ when it does not cause confusion. If $\hh_A$ and $\hh_B$ are two spaces and $A$ is an operator on $\hh_A$, we shall denote by $A$ the operator $A\otimes I_B$ on $\hh_A\otimes \hh_B$.
\item A vector $v\in \hh$ may also be written $\ket{v}$, and the corresponding linear form is denoted $\bra{v}$, so that $\ket{v}\bra{v}$ is the orthogonal projection on $\C v$. In any tensor space $\hh_A\otimes \hh_B$, The partial trace with respect to $\hh_B$ is written $\tra_B$ or $\tra_{\hh_B}$.
\item The algebra of bounded operators on $\hh$ is written $\bb(\hh)$, endowed with the operator norm $\norm{A}$ (sometimes written $\norm{A}_\infty$ to avoid confusion with other norms). The space of compact operators on $\hh$ is written $\bb^\infty_0(\hh)$. An operator on $\bb(\hh)$ is called a super-operator.
\item The adjoint of an operator $A$ is denoted $A^*$.
\item The Schatten space of order $p$ is the space $\cals^p(\hh)$ of bounded operators $A$ such that $\tr{\abs{A}^p}<+\infty$, endowed with the norm $\norm{A}_p=\tr{\abs{A}^p}^{1/p}$. In particular, $\cals^1(\hh)$ is the space of trace-class operators.
\item The $\sigma$-weak (or ultraweak) topology on $\bb(\hh)$ is the topology generated by the seminorms 
\[
\norm{A}_{(u_i)_{i\in \N}, (v_i)_{i\in \N}}=\sum_{i\in\N} \scal{u_i,\,Av_i}
\]
where the $u_i$ and $v_i$'s are vectors in $\hh$ with $\sum_{i\in \N} \norm{u_i}^2+\norm{v_i}^2<+\infty$.
\item For a measured space $(\xx, \ff, \mu)$ we write the corresponding $L^p$ space as $L^p(\xx, \ff, \mu)$ or when it does not cause confusion $L^p(\xx, \mu)$ or even $L^p(\xx)$.
\item For any Banach space $B$, we write $L^2(\xx, B, \mu)$ the space of $L^2$ function from $\xx$ to $B$, and the Sobolev space of functions $f:\R^n\rightarrow B$ with distributional derivatives $f^{(k)}\in L^p$ for $k<l$ is written $W^{l, p}(\R^n, B)$. For $p=2$ and $B=\hh$ a Hilbert space, it is itself a Hilbert space and is written $H^l(\R^n, \hh)$. It is isomorphic to $\hh\otimes H^l(\R^n)$ and injected to a dense subset of $L^2(\R^n, \hh)=\hh\otimes L^2(\R^n)$. We write $X$ the position operator defined by $Xf(x)=xf(x)$, and $P=-i\partial_x$ the impulsion operator with domain $H^1(\R, Leb)$.
\item On the space $L^2(\xx, \mu)$, for any measurable function $f : \xx \rightarrow \C$ we write $M_f$ the operator of multiplication by $f$, defined by $M_f g(x)=f(x)g(x)$ for any $g$ such that $fg \in L^2(\xx, \mu)$. 
\item We write $\Un_A$ the indicator function of the set $A$, and $\Un=\Un_\xx$.
\item We write $\otimes_{alg}$ the algebraic tensor product and $\otimes$ the completed tensor product of Hilbert spaces.
\item We generally use the letter $\calg$ for isometries or unitary operators whose role is to identify a space as the subspace of another, or to identify two representations of the same space. This type or map is often implicit in the literature of quantum mechanics, so there is no standard notation; we chose the letter $\calg$ because it evoques the curved arrow $\hookrightarrow$ used for injections in category theory. 
\end{itemize}

\subsubsection{Unitary and Open Quantum Walks}\label{subsub:oqw}

Unitary quantum walks are generally called simply \enquote{quantum random walks}, we add \enquote{unitary} to distinguish them from open quantum walks. They were formally introduced in \cite{Aharonov_Davidovich_Zagury_93_quantum_random_walks} as quantum version of classical random walks on graphs, and they were extensively studied, notably in relation to quantum computing: universal quantum computation can be obtained with UQW \cite{universal_quantum_computation_10} and has been used to develop quantum algorithm, generally for the search of marked node in graphs	(see \cite{quantum_walks_spatial_search_16}\cite{quantum_walks_tree_17} among many other articles). See the comprehensive review \cite{quantum_walks_review_12} on unitary quantum random walks. 

For the sake of completeness, let us briefly describe unitary quantum walks. A UQW represents a quantum particle moving on a graph $G=(\vv, E)$, where the set of vertices $\vv$ is countable or finite. The internal state of the particle is described by a space $\hh_G$ (which is called the \enquote{gyroscope} in this article; it is also called the \enquote{quantum coin}, the internal space or the chirality space in the literature). The Hilbert space of the position of the particle is $\hh_z=L^2(\vv, \nu)$ where $\nu$ is the counting measure on $\vv$. We write $(x\rightarrow y)$ or $(y\leftarrow x)$ an edge oriented from $x$ to $y$ and $x_-=\set{y\in \vv|(y\rightarrow x)\in E}$ and $x_+=\set{y\in \vv| (x\rightarrow y) \in E}$. The definition of an unitary quantum walk is the following: 

\begin{defi}
A unitary quantum walk on $G$ with gyroscope $\hh_G$ is the quantum dynamics represented by a unitary operator $U$ on $\hh_G\otimes \hh_z$ which is of the form
\begin{align}\label{eq:defi_uqw}
U=\sum_{(y\leftarrow x) \in E} U_{(y\leftarrow x)}\otimes \ket{y}\bra{x} ~
\end{align}
where the $U_{(y\leftarrow x)}$'s are bounded operators on $\hh_G$. 
\end{defi}

An operator of the form of Equation \eqref{eq:defi_uqw} is unitary if and only if for any $x,y\in \vv$ we have
\[
\sum_{z \in x_+\cap y_+} U_{(x\leftarrow z)}^* U_{(y\leftarrow z)}=\delta_{x,y} \Un_{\hh_G}~.
\]
The most classical example is the one of translation-invariant UQW on the graph $\Z$ (with edges between nearest neighbours). It can always be written of the form
\[
U= B_-\otimes D^* + B_+\otimes D
\]
where $D$ is the translation to the right
\[
D=\sum_{i \in \Z} \ket{i+1}\bra{i}
\]
and $B_-$ and $B_+$ are operators such that $B_-^* B_-+B_+^* B_+=I_G$ and $B_-^*B_+=0$ (or equivalently if $\hh_G$ is of finite dimension, there exists an orthogonal projector $P$ and a unitary $V$ on $\hh_G$ such that $B_-=VP$ and $B_+=V(I_G-P)$).
\vspace{0.5cm}

Open quantum Walks where introduced in \cite{OQWbirth} as another dynamics on $\hh_G\otimes \hh_z$, which is not unitary as for UQWs but completely positive, meaning that it corresponds to the dynamics of an open quantum system. Let us briefly describe this concept. We look at a system described by a Hilbert space $\hh_S$ in interaction with an exterior system $\hh_p$, which are initially independent (that is in a state of the form $\rho_S\otimes \rho_B$) and evolve during some time $\tau$, with an evolution described by a unitary $V$. The state on $\hh_S$ after the evolution is then described by $\tra_{\hh_p} \left(U (\rho_S\otimes \rho_B) U^* \right)$. This leads to the following definition:
\begin{defi}\label{def:channel}
We call \emph{quantum channel} on $\hh_S$ a linear map $\Lambda$ on $\gs(\hh_S)$ which is of the form
\[
\Lambda(\rho)=\tra_{\hh_p}\left(V (\rho\otimes \rho_p)V^*\right)
\]
for some space $\hh_p$ and some state $\rho_p$ on $\hh_p$ and some unitary $V$ on $\hh_S\otimes \hh_p$. 
\end{defi}

Quantum channels can be characterized as the completely positive, trace-preserving and $\sigma$-weakly continuous maps on bounded operators (see for example Chapter 6 of \cite{AttalLecture}). Alternately, they are the maps which are of the form
\[
\Lambda(\rho)=\sum_{k=1}^{r} K_k \rho K_k^*
\] 
where $r\in \N\cup \set{+\infty}$, the $K_k$'s are bounded operators on $\hh_A$ with $\sum_{k=1}^{r} K_k^*K_k=I_A$ and are called the Krauss operators of $\Phi$.  The tripe $(\hh_p, V, \rho_p)$ corresponding to $\Lambda$ is called a Stinespring dilation of the channel (it is not unique).

We can now define open quantum walks:

\begin{defi}
An Open Quantum Walk (OQW) on the graph $G=(\vv, E)$ with gyroscope space $\hh_G$ is a quantum channel $\Lambda$ on $\hh_G\otimes \hh_z$ which is of the form
\[
\Lambda(\rho)=\sum_{(y\leftarrow x) \in E}  \Big(K_{(y\leftarrow x)}\otimes \ket{y}\bra{x}\Big)\rho\left( K_{(y\leftarrow x)}^*\otimes \ket{x}\bra{y}\right)~
\]
for some operators $K_{(y\leftarrow x)}$ satisfying
\[
\sum_{y\in x_+} K_{(y\leftarrow x)}^* K_{(y\leftarrow x)}=I_G~
\]
for all $x\in \vv$.
\end{defi}

A peculiar feature of OQWs is that $\Lambda(\rho)$ is always block-diagonal with respect to the basis $(\ket{x})_{x\in \vv}$ of $\hh_z$, that is, we can write
\[
\Lambda(\rho)=\sum_{x\in \vv} \rho_x\otimes \ket{x}\bra{x}
\]
for some familly of positive semi-definite operators $(\rho_x)_{x\in \vv}$. 

To any OQW corresponds a stochastic process called the quantum trajectories of the OQW:

\begin{defi}
The quantum trajectory of an OQW $\Lambda$ is the process $(X_n, \rrho_n)_{n\in \N}$ with $X_n\in \vv$ and $\rrho_n \in \gs(\hh_G)$ such that $(X_{n+1}\leftarrow X_n)\in E$ for all $n$ and with the following transition probabilities:
\begin{align}\label{eq:defi_traj_oqw}
\P\left(X_{n+1}=y~|~X_n=x\right)&=\tr{K_{(y\leftarrow x)}\rrho_n K_{(y\leftarrow x)}^*} \\
\rrho_{n+1}&=\frac{K_{(X_{n+1}\leftarrow X_n)}\rrho_n K_{(X_{n+1}\leftarrow X_n)}^*}{\tr{K_{(X_{n+1}\leftarrow X_n)}\rrho_n K_{(X_{n+1}\leftarrow X_n)}^*}}~.
\end{align}

\end{defi}

When we fix for initial state $X_0=x$ and $\rrho_0=\rho$ the quantum trajectory is related to the OQW by the formula
\[
\E\left(\rrho_n \otimes \ket{X_n}\bra{X_n}\right)=\Lambda^n(\rho\otimes \ket{x}\bra{x})~.
\]

This direct relation between the OQW and a random walk on the graph makes it closer to classical random walks than UQW. The concept of OQW has attracted significant interest; a central limit Theorem on the trajectories of translation invariant OQW on $\Z^n$ has been proved in \cite{OQW_CLT} and extended to more general lattices in \cite{OQW_CP14_lattice}, \cite{OQW_CKSY18_crystal} and completed by a large deviation principles in \cite{OQW_CP14_lattice}, while criterions for the ergodic properties of the quantum channel $\Lambda$ where proved in \cite{OQW_CP15_ergodic}. Two notable generalisation of OQW have been defined, one which interpolates between OQW and UQW \cite{OQW_XY12_coin}, and another which considers continuous-time OQW \cite{OQW_P14_continuous_time}, still on discrete graphs. 

\subsubsection{The formal definition of the Open Quantum Brownian Motion}

The idea of the Open Quantum Brownian Motion is to define a dynamics which is similar to the OQW dynamic, but in continuous time and continuous space (with a particle moving on the line). It is defined as a limit of a family of OQW on the graph $\Z$, with a time scale $\tau$ and a space scale $\delta=\sqrt{\tau}$ going to zero. {\bf In the rest of the article, we will write $\delta=\sqrt{\tau}$ the space scale.} Let us define formally the Open Quantum Brownian Motion, following \cite{OQBM}. we consider the graph $\delta \Z$ (with nearest neighbours edges) and $\hh_{\tau,z}=l^2(\delta \Z)$ and a gyroscope space $\hh_G$. We define the OQW $\Lambda_{\tau}$ on $\gs(\hh_G\otimes \hh_{\tau, z})$ by
\begin{align}\label{eq:discreteOQBM_OQW}
\Lambda_\tau(\rho)=\sum_{x\in \delta\Z} \left(B_{\tau, -} \otimes \ket{x-\delta}\bra{x} \right)\rho\left(B_{\tau, -}^*\otimes \ket{x}\bra{x-\delta}\right) +\left(B_{\tau, +} \otimes \ket{x+\delta}\bra{x} \right)\rho\left(B_{\tau, +}^*\otimes \ket{x}\bra{x+\delta}\right)
\end{align}

where the Krauss operators $B_{\tau, +}$ and $B_{\tau, -}$ satisfy
\begin{align}\label{eq:defi_B}
B_{\tau,\pm 1}&=\frac{1}{\sqrt{2}}\left(I\pm\delta N+\tau\left(-iH-\frac{1}{2} N^*N\pm M \right)\right)+O(\tau^{3/2})~.
\end{align}
for some bounded operators $N, H, M$ on $\hh_G$ with $H$ self-adjoint. It is argued in \cite{OQBM} that it is the only choice of $B_{\tau, \pm}$ such that $\Lambda_\tau^{\ent{t/\tau}}$ converges for all $t$ as $\tau=\delta^2 \rightarrow 0$. Let us derive formally the limit: consider the state $\rho(n)=\Lambda^n_\tau(\rho)$. For any $n>0$ it is of the form
\[
\rho(n)=\sum_{x\in \delta \Z} \rho(n, x)\otimes \ket{x}\bra{x}
\]
for some positive semi-definite operators $\rho(n,x)$ on $\hh_G$. By the definition of $\Lambda_\tau$ we have
\begin{align}\label{eq:estimate_rho_n_x}
\begin{split}
\rho(n+1,x)&=\frac{\rho(n, x+\delta)+\rho(n, x-\delta)}{2}\\
&+\delta\left( N \frac{\rho(n, x+\delta)-\rho(n, x-\delta)}{2}+\frac{\rho(n, x+\delta)-\rho(n, x-\delta)}{2}N\right)\\
&+\tau \left(\call\left(\frac{\rho(n, x+\delta)+\rho(n,x-\delta)}{2}\right)+M\frac{\rho(n,x+\delta)-\rho(n, x-\delta)}{2}+\frac{\rho(n,x+\delta)+\rho(n,x-\delta)}{2}M^* \right)+O(\tau \sqrt{\tau})
\end{split}
\end{align}
where the super-operator $\call$ on $\gs(\hh_G)$ is defined by 
\begin{align}\label{eq:defi_call}
\call(\rho)=-i[H, \rho]-\frac{1}{2}\set{N^*N, \rho}+N\rho N^*~.
\end{align}

Assume that $(t,x)\mapsto \rho(\ent{t/\tau}, x)$ converges as $\tau\rightarrow +\infty$ to some function $(t,x) \mapsto \rho(t,x)$ from $[0, +\infty)\times \R$ to $\cals^1(\hh_G)$. We have the formal, non-rigorous estimates
\begin{align}
\frac{\rho(t+\tau, x)-\rho(t,x)}{\tau}&\simeq \frac{\partial}{\partial t} \rho(t, x) \\
\frac{\rho(t, x+\delta)-\rho(t, x-\delta)}{2 \delta}&\simeq \frac{\partial}{\partial x} \rho(t,x) & \frac{\rho(t, x+\delta)+\rho(t,x-\delta)-2\rho(t,x)}{\delta^2}&\simeq \frac{\partial^2}{\partial x^2}\rho(t,x) \\
\end{align}

Assuming that these estimates are justified and replacing them in Equation \eqref{eq:estimate_rho_n_x} we obtain
\begin{align}\label{eq:LindbladOQBM}
\frac{\partial}{\partial t} \rho(t, x)&= \frac{1}{2} \frac{\partial^2}{\partial x^2} \rho(t,x)+ N \frac{\partial}{\partial x} \rho(t,x)+\frac{\partial}{\partial x} \rho(t,x)N^* +\call(\rho(t, x))~.
\end{align}

This equation defines the dynamics of the Open Quantum Brownian Motion. We call it the Lindblad Equation for the OQBM as it represents the generator of a continuous-time dynamics $t\mapsto \Lambda^t(\rho)$ where $\Lambda^t$ is a quantum channel for all $t\in [0, +\infty)$.  This raises several problems, but before listing them let us describe the two other descriptions of the OQBM. The first is the stochastic process of quantum trajectories. Let $(X_n, \rrho_n)_{n\in \N}$ be the quantum trajectories of the Open Quantum Walk $\Lambda_\tau$ (with $X_n\in \delta \Z$), then another formal estimate gives a stochastic differential equation for the limit of $(X_{\ent{t/\tau}}, \rrho_{\ent{t/\tau}})_{t\in [0, +\infty)}$ as $\tau\rightarrow 0$:
\begin{align}\label{eq:belavkinOQBM}
\left\{
\begin{array}{ll}
d\rrho_t&=\call(\rrho_t)dt+\left(N\rrho_t+\rrho_t N^*-\rrho_t\Tau(\rho_t)\right)dB_t\\
dX_t&=\Tau(\rrho_t)dt+dB_t
\end{array}
\right.
\end{align}
where $B_t$ is a Brownian motion and $dB_t$ is its It\^o differential, and $\Tau(\rho)=\tr{(N+N^*)\rho}$.

The last representation of the OQBM is the \enquote{dilated} one, it consists in a unitary evolution $(\gu_t)_{t\in [0, +\infty)}$ on a space $\hh_G\otimes \hh_z\otimes \Phi$ where $\Phi$ is the bosonic Fock space on $L^2(\R)$. It satisfies a Hudson-Parthasaraty Equation (whose formalism is introduced later in the article):
\begin{align}\label{eq:HPOQBM}
d\gu_t=\Big((-iH-\frac{1}{2}N^*N+\frac{1}{2}\partial_x^2-\partial_x N)dt+(N-\partial_x)da^0_1(t)+(-N^*-\partial_x)da^1_0(t) \Big) \gu_t~
\end{align}
where $a^1_0(t)$ is the creation operator of $\Un_{[0, t)}$ and $a_1^0(t)=\left(a^1_0(t)\right)^*$ is the corresponding annihilation operator. It is related to the Lindblad dynamics of the OQBM through the partial trace, 
\[
\Lambda^t(\rho)=\tra_{\Phi}\left( \gu_t \left(\rho\otimes \ket{\Omega}\bra{\Omega}\right) \gu_t^*\right)~
\]
and it is related with the quantum trajectories through the concept of continual measurement described in the third section of this article.

The purpose of this article is to adress the many mathematical problems aroused by these definitions, as listed below.

\begin{enumerate}
\item A first problem is the projection on diagonal states. Indeed, for any $t\geq\tau$ the state $\rho_{\tau,t}=\Lambda^{\ent{t/\tau}}(\rho)$ is of the form
\[
\sum_{x\in \delta \Z} \rho(t,x)\otimes \ket{x}\bra{x}
\] 
i.e. it is in the algebra $\bb(\hh_G)\otimes L^\infty(\delta \Z)$. Likewise, in the limit $\tau\rightarrow 0$, for any $t>0$ the state at time $t$ should be in the algebra $\bb(\hh_G)\otimes L^\infty(\R)$, which we may write in the formalism of spectral measure as
\[
\rho(t)=\int_{x\in \R} \rho(t,x)d\ket{x}\bra{x}~.
\]
It means that it operates on $u\otimes f \in\hh_G\otimes L^2(\R)=L^2(\R, \hh_G)$ as
\[
\rho(t) \left(u\otimes f\right)(x)=f(x) \rho(t,x) u ~.
\]
 But such an operator \emph{cannot be a state}, since it is not even trace-class when it is nonzero (indeed, either it has a continuous spectrum, either it has nonzero eigenvalues with infinite-dimensional eigenspace). Thus, it is impossible for $\Lambda_\tau^\ent{t/\tau}$ to converge on the full algebra $\bb(\hh_G\otimes \hh_z)$. This problem is adressed in this article the following way: at first, we do not consider the OQW map $\Lambda_\tau$ but another map $\tilde{\Lambda}_\tau$ which coincides with $\Lambda_\tau$ on the algebra of diagonal states $\bb(\hh_G)\otimes L^\infty(\delta\Z)$. It is defined by
\[
\tilde{\Lambda}_\tau(\rho)=\left(B_{\tau, -} \otimes D_\tau^*\right)\rho\left(B_{\tau, -}^*\otimes D_\tau\right)+\left(B_{\tau, +} \otimes D_\tau\right)\rho\left(B_{\tau, +}^*\otimes D_\tau^*\right)
\]
where $D_\tau=\sum_{x\in \delta\Z} \ket{x+\delta}\bra{x}$ is the translation by $\delta$ to the right.
We show that $\tilde{\Lambda}_\tau^{\ent{t/\tau}}$ converges strongly to a map $\Lambda^t$, and that $\left(\Lambda^t\right)^*$ preserves the algebra $\bb(\hh_G)\otimes L^\infty(\R)$. Thus, we may consider states \emph{on} the algebra $\bb(\hh_G)\otimes L^\infty(\R)$, for which $\Lambda_\tau$ and $\tilde{\Lambda}_\tau$ coincide, and converges.
\item The convergence for the quantum trajectories is the less problematic. We prove it in the case where $\hh_G$ is of finite dimension, on the time interval $[0, T]$ for some fixed $T$, as a convergence in distribution in the Skorokhod space of continuous functions (Proposition \ref{prop:trajectoryOQBM}). The convergence is a direct consequence of a theorem of Pellegrini \cite{PellegriniDiffusive08}.
\item For the convergence to the unitary dilation $\gu_t$, we construct a unitary dilation $(\gu_{\tau,n})_{n\in \N}$ of the discrete-time semigroup $(\tilde{\Lambda}_\tau^n)_{n\in \N}$, and show its strong convergence to a unitary $\gu_t$ satisfying the Hudson-Parthasarathy Equation \eqref{eq:HPOQBM} (Theorem \ref{theo:HPOQBM}).
For this, we use a theorem of Attal and Pautrat \cite{AttalPautrat}; this theorem is designed to work with bounded operators, while the operator $\partial_x$ is unbounded. This problem is bypassed by considering the restriction to the space $\dd_C=\set{f \in L^2(\R) | \mbox{supp}(\ff f) \subset [-C, C]}$ where $\ff$ is the Fourier transform, and using the density of the space $\cup_{C>0} \dd_C$. 
\item We show the strong convergence of $\tilde{\Lambda}_\tau^{\ent{t/\tau}}$ to the quantum channel 
\[
\Lambda^t(\rho)=\tra_{\Phi}\left(\gu_t\left(\bullet\otimes \ket{\Omega}\bra{\Omega}\right) \gu_t\right)~.
\]
 The problem is then to show that $\Lambda^t(\rho)$ satisfies indeed Equation \eqref{eq:LindbladOQBM}, provided $\rho$ is a state on $\bb(\hh_G)\otimes L^\infty(\R)$ and is sufficiently regular. We define the space of Sobolev states, show an extended version of Equation \eqref{eq:LindbladOQBM} for Sobolev states on $\bb(\hh_G\otimes \hh_z)$ (Theorem \ref{theo:LindbladOQBM1}) with the use of the quantum stochastic calculus on $\gu_t$, and restrict this equation to states on $\bb(\hh_G\otimes \hh_z)$ to obtain Equation \eqref{eq:LindbladOQBM} (Theorem \ref{theo:LindbladOQBM2}).
\item One conceptual problem is to relate quantum trajectories and the dilation $\gu_t$. This is the object of the third section of this article, where we expose the formalism of continual measurement of non-demolition evolution. We prove a partial result relating the convergence of quantum trajectories and of the dilation; this theorem is redundant in the case of the OQBM since the convergence of quantum trajectories can be proved by other ways, but it applies to more general evolution under continual measurement.
\item The OQBM can be generalized in various ways, as noted in \cite{OQBM}, and related problems are listed at the end of this article.
\end{enumerate}

\section{The Open Quantum Brownian motion}

 In this section, we outline some of the mathematical objects of quantum mechanics, describing states on some von Neumann algebras, measurement, the repeated interactions setup and the Belavkin Equation. We show that the discrete OQBM can be seen as a repeated interactions model supplemented with a quantum description of a pointer linked with some repeated measurement.

\subsection{von Neumann algebras and quantum states}

The notion of standard measured space is crucial in the mathematical definition of measure, since for every Hilbert space $\hh$ there exists a standard measured space $(\xx, \ff, \nu)$ such that $\hh$ is isomorphic to $L^2(\xx, \nu)$. This also allows to study commutative von Neumann algebras, and to relate the notion of quantum state to classical probabilities and the measurement of observables.

 \subsubsection{Standard measured space}\label{subsub:meas}

Standard measured spaces form a very large class of measured space; notably, two spaces of special interest in this article are $\xx=\R$ with the Lebesgue measure, and $\xx=\calw([0,+\infty))$ the Wiener space on $[0, +\infty)$ equiped with the Wiener measure (i.e. the space of continuous functions on $[0, +\infty)$ equipped with the measure corresponding to the Brownian motion). Standard measured spaces have many different characterizations, see the chapter on Lebesgue-Rohlin spaces in Bogachev II \cite{Bogachev}; let us describe two of them:
\begin{defi}
Let $(\xx, \ff, \mu)$ be a measured space (every measured are nonnegative in this article). It is called a \emph{standard measured space} if it satisfies one of the following equivalent properties:
\begin{enumerate}
\item There exists a measure $\nu$ on $\R$ of the form $\nu=\nu_1+\sum_{i\in \N} c_i \delta_i$ where $\nu_1$ is absolutely continuous with respect to the Lebesgue measure, the $\delta_i$ are the Dirac distributions at $i$ and the $c_i$ are nonnegative numbers, such that $(\xx, \ff, \mu)$ is almost isomorphic to $(\R,\bb(\R),\nu)$, that is, there exists sets of full measure $A\subset \xx$ and $B\subset \R$ and a measure-preserving isomorphism between $(A, \mu)$ and $(B, \nu)$.
\item There exists a complete metric $d$ on a set of full measure $D\subset \xx$ such that $\ff|_D$ is the $\sigma$-algebra generated by open sets for $d$ and $\mu$ is a Radon measure for this topology.
\end{enumerate}
\end{defi}

Note that standard measured spaces are necessarily almost separated (i.e. for almost every $x\neq y \in \xx$ there exists two disjoint measurable sets $A,B\in \ff$ with $x\in A$ and $y\in B$). More importantly, if $\ff_1\subset \ff$ is another $\sigma$-algebra, the measured space $(\xx, \ff_1, \mu)$ is standard if and only if $\ff_1=\ff$. If $\ff_1\neq \ff$, we make $(\xx, \ff_1, \mu)$ into a standard probability space by quotient: 

\begin{defi}
For any standard measured space $(\xx, \ff, \mu)$ with a sub-$\sigma$-algebra $\ff_1$, let $\xx/\ff_1$ the quotient of $\xx$ by the relation: $x\sim y$ if every set $A\in \ff_1$ containing $x$ also contains $y$. There is a surjective map $s_{\ff_1}:\xx\rightarrow \xx/\ff_1$, we endow $\xx/\ff_1$ with the image of $\ff_1$ by $s_{\ff_1}$ and the push-forward measure of $\mu$ by $s_{\ff_1}$, which we still write $\ff_1$ and $\mu$. The space $(\xx/\ff_1, \ff_1, \mu)$ is a standard measured space, called the \emph{quotient} of $(\xx, \ff, \mu)$ by $\ff_1$. 

There exists many different maps $r_{\ff_1}: \xx/\ff_1\rightarrow \xx$ such that $s_{\ff_1}\circ r_{\ff_1}=I_{\xx_1}$. 
Each of them gives an identification of $\xx_1$ with a subspace of $\xx$, and we have a map $c=r_{\ff_1}\circ s_{\ff_1}: \xx\rightarrow \xx$ onto this subspace.

An \emph{extension} of a standard measured space $(\xx_1, \ff_1, \mu_1)$ is another standard measured space $(\xx, \ff, \mu)$ with a surjective measurable map $s:\xx\rightarrow \xx_1$ such that the push forward measure $s_*\mu$ of $\mu$ by $s$ is $\mu_1$.
\end{defi}

These notions are useful in the description of commutative von Neumann algebras.

\subsubsection{Commutative von Neumann algebras}\label{subsub:von}

The set of quantum observables of a system is described by a von Neumann algebra on $\hh$, i.e. a unital subalgebra of $\bb(\hh)$ which is stable by adjoint and closed for the strong topology.
 This article does not involve most of the subtleties of von Neumann algebra theory, since we are essentially interested in the simplest cases: the full algebra $\bb(\hh)$, the commutative von Neumann algebras and the tensor products of these. Let us recall a few facts about commutative von Neumann algebras:
 
 \begin{enumerate}
 \item For any standard probability space $(\xx, \ff, \mu)$ and any sub-$\sigma$-algebra $\ff_1\subset \ff$ the space $L^\infty(\xx,\ff_1, \mu)$ is identified with a commutative von Neumann algebra on $L^2(\xx, \ff, \mu)$ by $f\in L^\infty(\xx, \ff_1, \mu) \mapsto M_f$ (the operator of multiplication by $f$).
 \item Let $\cala\subset \bb(\hh)$ be a commutative von Neumann algebra. Then there exists a standard measured space $(\xx, \ff, \mu)$, a sub-$\sigma$-algebra $\ff_1\subset\ff$ and a unitary operator $\pi$ from $L^2(\xx, \ff, \mu)$ to $\hh$ such that $\cala=\pi^*~L^\infty(\xx, \ff_1, \mu)~\pi$.  Thus, if we consider the quotient $\xx_1=\xx/\ff_1$, then $\cala$ is isomorphic (as a $C^*$-algebra) to $L^\infty(\xx_1, \ff_1, \mu)$. The algebra $\cala$ is a maximal commutative von Neumann algebra if and only if $\ff_1=\ff$ (up to measure-zero sets). It is called \enquote{discrete} if $\xx_1$ is countable or finite, the $\sigma$-algebra $\ff_1$ is then called \enquote{coarse}\footnote{the term \enquote{discrete} $\sigma$-algebra often refers to the $\sigma$-algebra of all subsets of $\xx$, so we use coarse to avoid confusion}.
\item Let $\cala_1 \subset \cala_2$ be two commutative von Neumann algebras on a von Neumann algebra with two isomorphisms of $C^*$-algebras $\psi_1: \cala_1\rightarrow L^\infty(\xx_1, \ff_1, \mu_1)$ and $\psi_2:\cala_2 \rightarrow L^\infty(\xx_2, \ff_2, \mu_2)$. Then there exists a measurable map $\eta: \xx_2 \rightarrow \xx_1$ such that $\mu_1$ is absolutely continuous with respect to the push forward measure $\eta_*\mu_2$ and for any $f\in L^\infty(\xx_1, \ff_1, \mu_1)$ we have $\psi_2\circ \psi_1^{-1}(f)=f\circ \eta$.
 \end{enumerate}

See Takesaki's book \cite{TakesakiBook}, notably Theorem 8.21 and Lemma 8.22. An application of the last fact is that if $U$ is an isometry of $\hh$ with $U \cala_1 U^* \subset \cala_2$ then its action on $\cala_1$ can be implemented by some map $\eta$ between the underlying spaces $\xx_1$ and $\xx_2$.  

  A full study of a non-maximal commutative von Neumann algebra involves direct integrals of Hilbert spaces. We don't need it here, so let us just give a taste of it: if $\cala\simeq L^\infty(\xx_1, \ff_1, \mu)$ then we can decompose $\hh$ as $\int_{\xx_1}^\oplus \hh(x)d\mu(x)$ where $x\mapsto \hh(x)$ is a measurable field of Hilbert spaces, and the elements of $\cala$ are operators of the form $\int_\xx^\oplus f(x) I_{\hh(x)} d\mu(x)$.

\subsubsection{Quantum states}\label{subsub:states}

Let us describe states on several von Neumann algebras. The state of a quantum system with observables in a von Neumann algebra $\cala$ is modeled the following way: 
\begin{defi}
A (normal) state on a von Neumann algebra $\mm$ is a linear form $\rho$ on $\mm$ which is:
\begin{itemize}
\item positive, i.e. $\rho(A)\geq 0$ for any positive semi-definite operator $A \in \mm$.
\item normed, i.e. $\rho(I)=1$
\item normal, i.e. continuous for the $\sigma$-weak topology, or equivalently  for any sequence of mutually orthogonal projections $(P_n)_{n\in \N}\in \mm^\N$ we have $\sum_{n\in \N} \rho(P_n)=\rho(\sum_{n\in \N} P_n)$.
\end{itemize} 
The set of states on $\mm$ is written $\gs(\mm)$ or simply $\gs(\hh)$ if $\mm=\bb(\hh)$.
\end{defi}
Let us consider the two cases of maximal commutative von Neumann algebra and of the full von Neumann algebra:

{\bf States on $\cala=L^\infty(\xx,\ff, \mu)$: } any state $\rho$ on $\cala$ is of the form
\[
\rho(f)=\int_\xx f(x) p_\rho(x) d\mu(x)
\]
where $p$ is a positive function on $\xx$ with $\int_\xx p(x)d\mu(x)=1$. Hence the set $\gs(L^\infty(\xx, \mu))$ can be identified with the set of probability measures which are absolutely continuous with respect to $\mu$. 

{\bf States on $\bb(\hh)$:} any state $\rho$ on the full algebra is the form
\[
\rho(A)=\tr{AT_\rho}
\]
where $T_\rho$ is a positive semi-definite trace-class operator on $\hh$ with $\tr{T_\rho}=1$. By convention, we use the letter $\rho$ for both the state and  the corresponding trace-class operator, and we identify the set $\gs(\hh)$ with the set of positive semi-definite trace-class operators of trace 1.

{\bf States on $\bb(\hh)\otimes L^\infty(\xx, \ff, \mu)$: } This is the mix of the two previous situations: a state $\rho$ on $\bb(\hh)\otimes L^\infty(\xx, \ff, \mu)\subset \bb(\hh\otimes L^2(\xx,\ff, \mu))$ is of the form 
\[
\rho(A\otimes f)=\int_\xx \tr{A Q_\rho(x)}f(x) d\mu(x)
\]
where $x\mapsto Q_\rho(x)$ is a measurable function from $\xx$ to the set of positive semi-definite trace-class operators on $\hh$ such that $\int_\xx \tr{Q_\rho(x)}d\mu(x)=1$. We call $Q_\rho$ the density matrix function. 
\vspace{0.5cm}

\begin{rem} 
\begin{enumerate}
\item If $\mm_1\subset \mm_2$ are two von Neumann algebras, we may extends states on $\mm_1$ to states on $\mm_2$, and restrict states on $\mm_2$ to states on $\mm_1$. In particular, if $\mm_1=L^\infty(\xx, \mu)$ and $\mm_2=\bb(L^2(\xx, \mu))$, a state on $\mm_1$ can be extended in many different ways to a state on $\mm_2$, notably we can make it a pure state: take $f=\sqrt{p}$ where $p$ is the probability density of the state with respect to $\mu$, and consider the state $\ket{f}\bra{f}$ on $\mm_2$. We may also be tempted to take the multiplication operator $M_p$ as another extension, but this operator may not be trace class when $\xx$ is not coarse.

\item Another important example is the case of a bipartite system. If $\hh=\hh_A\otimes\hh_B$ and we are given a state $\rho$ on $\mm_2=\bb(\hh)$, its restriction to $\mm_1=\bb(\hh_A)\otimes \set{I_B}$ has for density matrix the partial trace of $\rho$ with respect to $B$, that is $\rho_B=\tra_B(\rho)$. 

\item With $\hh_B=L^2(\xx, \ff, \mu)$ and $\mm_1=\bb(\hh_A)\otimes L^\infty(\xx, \ff, \mu)$ and $\mm_2=\bb(\hh_A)\otimes \bb(\hh_B)$ the situation is more subtle. A state $\rho$ on $\mm_2$ can always be described by a kernel $(x, y)\mapsto K_\rho(x, y)$ from $\xx\times \xx$ to $\cals^1(\hh)$, such that for any function $f\in L^2(\xx, \hh_A)=\hh_A\otimes L^2(\xx, \ff, \mu)$ we have
\[
(\rho f)(x)=\int_{\xx} K_\rho(x, y) f(y)d\mu(y)
\]
(where we see $\rho$ as an operator on $\hh$). To describe the state $\rho_{\mm_1}$ on $\mm_1$ it seems natural to take for density matrix function $Q_{\rho_{\mm_1}}(x)=K_\rho(x,x)/\tr{K_\rho(x,x)}$.
Unless $K$ is continuous with respect to some metric, this requires technicalities since the diagonal $\{(x,x) | x \in \xx\}$ is possibly of measure zero in $(\xx\times \xx, \mu\otimes \mu)$. This issue can be solved with the help of the Lebesgue differentiation theorem, either by averaging on small rectangles (see Brislawn \cite{brislawn91}) or with the notion of virtual continuity (see Vershik et al. \cite{vershik16}).
\end{enumerate}

\end{rem}

\subsubsection{Measure of an observable}\label{subsub:discretemeas}

Let $A$ be a self-adjoint operator on $\hh$ (which is not necessarily bounded). 
Assume that the system is in the state $\rho$. The measurement of $A$ is mathematically described the following way: the von Neumann algebra $\cala$ generated by $A$ is commutative, so there exists a unitary operator $\pi : \hh \rightarrow L^2(\xx, \mu)$ for some standard measured space $(\xx,\ff, \mu)$ and a measurable function $g : \xx \rightarrow \R$ such that $\pi^* A \pi=M_g$. 
Let $\rho$ be the state on the system, then $\pi^*\rho \pi$ restricts to a state on $L^\infty(\xx, \mu)$, that is, a probability measure $\P_\rho$ on $\xx$ which is absolutely continuous with respect to $\mu$. This makes $(\xx, \P_\rho)$ a probability space.
The result of the measurement is then the random variable $\tilde{A}_\rho$ on $(\xx, \P_\rho)$ defined by the function $g$.

Note that for a commuting family of self-adjoint operators $(A_\alpha)_{\alpha \in I}$ we can consider their joint spectral theory: there exists a unitary operator $U : \hh\rightarrow L^2(\xx, \mu)$ with $U^*A_\alpha U=M_{g_\alpha}$ for a family of functions $(g_\alpha)_{\alpha \in I}$. Thus, we can consider the family of random variables $\tilde{A}_{\alpha, \rho}$ on the same probability space $(\xx, \P_\rho)$. However, if $A$ and $B$ are not commuting, there is no consistent way to consider jointly $\tilde{A}_\rho$ and $\tilde{B}_\rho$ as random variables on the same probability space.

Now, it is not always possible to describe the quantum mechanical state of $\rho$ after the exact measurement. In the case where $A$ has only pure point spectrum, it is possible and we do it as follows.

\begin{defi}[State after the measurement]
Let $A$ be an observable of the form
\[
A=\sum_{a \in sp(A)} a P_a
\]
where the $P_a$ are mutually orthogonal projections. Write $\cala$ the commutative von Neumann algebra generated by $A$, it is isomorphic to $L^\infty(sp(A),~ \sum_a \delta_a)$. We endow $sp(A)$ with the probability $\P_\rho(a)=\tr{\rho P_a}$.  The state after the measurement of $A$ is the radom variable $\rho_{|\cala}$ on $(sp(A), \P)$ defined by
\[
\rho_{|\cala}(a)=\frac{P_a\rho P_a}{\tr{P_a\rho}}\,.
\] 
We may also write $\rho_{S|\cala}:=\tra_B(\rho_{|\cala})$, and to shorten notation we will often use the variant calligraphy $\rrho$ for a random density matrix corresponding to a deterministic density matrix $\rho$.
\end{defi}

The action of \emph{not reading the result of the measurement} consists in discarding the random variable $\tilde{A}_\rho$ and replacing $\rho_{|\cala}(a)$ by its expectancy $\rho'=\E(\rho_{|\cala})$. The operator $\rho'=\sum_{a\in \sp(A)} P_a \rho P_a$ is in $\gs(\hh)$. It carries all the information which can be obtained without the knowledge of $\tilde{A}_\rho$, since $\E\big(\tr{\rho_{|\cala} B}\big)=\tr{\rho'B}$ for any observable $B\in \bb(\hh)$. 
\newline

If $A$ has singular spectrum it is no more possible to describe the state after the exact measurement as a random variable on $\gs(\hh)$. For example, if we measure the position observable $X$ on $L^2(\R, Leb)$ the state of the system after the measurement should correspond to the Dirac measure $\delta_{\tilde{X}_\rho}$ on the algebra $L^\infty(\R, Leb)$, but it is not possible since states on this algebra are absolutely continuous with respect to the Lebesgue measure. This is linked to the fact that every repeatable instrument is discrete, see Ozawa \cite{ozawa_conditional_1985}.

This is not really a physical problem since no real-life measurement is exact, hence we only measure discrete observables in real life. Though, it is always better to have an idealization of the measure of continuous observables, and we show a way to circumvent these issues below.

\subsubsection{The quantum state after the measurement of a continuous observable and indirect measurement}\label{subsub:after_meas}

 This part is not used before Section \ref{sec:measure}, but it introduces the notion of \enquote{pointer unitary operator} which helps the understanding of the OQBM. The idea to describe the state after the exact measurement is to restrict the state to some subalgebra of $\bb(\hh)$. The case we consider is the following:
\begin{itemize}
\item The space $\hh$ is the tensor product of two Hilbert spaces $\hh_G$ and $\hh_B$. 
\item We want to measure a family of mutually commuting operators $(B_\alpha)_{\alpha \in I}$ acting on $\hh_B$. Write $\cala$ the von Neumann algebra generated by the $B_\alpha$'s.
\item We are interested on the state after the measurement on $\bb(\hh_G)$ only. It will be written $\rho_{G|\cala}$.
\end{itemize}
We will see that concentrating on the state on $\hh_G$ and ignoring the full picture on $\hh_G\otimes \hh_B$ allows us to get a rigorous definition of $\rho_{G|\cala}$. This setup is geared to describe indirect measurement.

Since the $B_\alpha$ are commuting, we can identify $\hh_B$ with $L^2(\xx, \mu)$ for some standard measured space space $(\xx, \ff, \mu)$ such that there exists measurable functions $g_\alpha$ with $B_\alpha=M_{g_\alpha}$. We want to define $\rho_{G|\cala}$ as a random variable with values in $\gs(\hh_G)$ on the probability space generated by the random variables $g_\alpha=\tilde{B}_{\alpha,\rho}$.

\begin{theo}\label{theo:random_rho}
Let $\rho$ be a state on $\bb(\hh_G\otimes L^2(\xx,\mu))$. Then there exists a measurable map $\ssig$ from $\xx$ to $\cals^1(\hh_B)$ such that for any $f\in L^\infty(\xx, \mu)$ and for any observable $A\in \bb(\hh_G)$ we have
\[
\tr{\rho~A\otimes M_f}=\int_\xx \tr{\ssig(x) A}f(x)d\mu(x)~. 
\]
It is unique (up to a $\mu$-negligible set), and $\ssig(x)$ it positive semi-definite and satisfies 
\[
\tr{\ssig(x)}=\frac{d\P_\rho}{d\mu}(x)
\]
 for $\mu$-ae $x$. It is called the \emph{unnormalized state} on $\hh_G$ associated to $(\xx, \mu)$. Note that its trace depends on the measure $\mu$ which is chosen. 

Now, consider a sub-$\sigma$-algebra $\ff_1\subset \ff$ and let $\cala=L^\infty(\xx, \ff_1, \mu)$. Let $\P_\rho$ the probability measure induced by $\rho$ on $\xx$. Then there exists a random variable $\rho_{G|\cala}$ on $(\xx, \ff_1, \P_\rho)$ with values in $\gs(\hh_G)$ such that for any operator $A \in \bb(\hh_G)$ and any random variable $f \in L^\infty(\xx,\ff_1,\P_\rho)$ we have
\[
\tr{\rho A\otimes M_f}=\E_{\rho}\left( \tra\big(\rho_{G|\cala}\,A\big) f\right)
\]
where on the right $f$ is seen as a random variable. The random variable $\rho_{G|\cala}$ is unique up to a set of probability zero, and for $\P$-almost a $x\in \xx$ we have $\rho_{G|\cala}(x)=\ssig(x)/p_\rho(x)$.

We will often write $\rrho$ for $\rho_{G|\cala}$ when it does not cause confusion, and we write $\ssig=u_{\xx}(\rho)$ (or $u_{(\xx, \mu)}(\rho)$ when the measure needs to be precised). 
\end{theo}
Note that $u_{\xx}: \rho\mapsto \ssig$ is an isometry, contrarily to the map $\rho\mapsto \rrho$.

\begin{proof}
The function $\ssig$ is just the matrix density function of the restriction of $\rho$ to $\mm=\bb(\hh_G)\otimes \cala$, so its existence is just a consequence of the Riesz theorem. 

We have $\int_\xx f(x) \tr{\ssig(x)}d\mu(x)=\rho(M_f)=\E_\rho(f)$ so $\tr{\ssig(x)}=dP_\rho/d\mu$, and so $\tr{\ssig(x)}$ is nonzero $\P_\rho$-almost surely. We now define
\[
R(x)=\frac{\ssig(x)}{\tr{\ssig(x)}}
\]
on $x$ such that $\ssig(x)\neq 0$. It is a random variable on $(\xx, \P_\rho)$. Now we take the conditional expectation with respect to the $\sigma$-algebra $\ff$ generated by the $g_\alpha$ on $\xx$~:
\[
 \rho_{G|\cala}=\E\big(R\,\big|\,\ff\big).
\]
It is easy to show that it fits the requirement of the theorem.

The uniqueness is straightforward.
\end{proof}

\begin{rem}
\begin{enumerate}
\item With this approach, we clearly separate the quantum superposition, described by a density matrix, and the classical randomness on the probability space $(\xx,\ff, \P_\rho)$.
 It is frequent in quantum filtering theory to define $\rrho$ as a state on the commutant of $\cala$, which is in general bigger than $\bb(\hh_B)\otimes L^\infty(\xx)$, but this does not define $\rrho$ explicitly as a random variable on some probability space. 
 \item  Note that the state $\rho_{G|\cala}$ contains more information that $\rho_G=\tra_{\hh_B}(\rho)$ since $ \rho_G=\E_{p_\rho d\mu}(\rho_{G|\cala})$. Thus, we have three descriptions of the state of the system, containing less and less information: the full state $\rho$ on $\hh_G\otimes\hh_B$, the random state $\rho_{G|\cala}$ and the state $\rho_G$. We could define a fourth description between $\rho_{G|\cala}$ and $\rho_G$ by using the theory of direct integral: if $\cala$ is the set of decomposable operators on $\hh=\int_{\xx}^\oplus \hh(x)d\mu(x)$ we may consider a random state $\rrho(X)$ on the random Hilbert space $\hh(X)$. This level of precision is not needed for our purpose.
\end{enumerate}
\end{rem}

As an application of this theorem, we can model the indirect measurement of an observable; it is a framework often called von Neumann measurement of an observable in the literature (\cite{Belavkin94},\cite{BushLahtiMittelstaedt}, \cite{GoughNotes}). Let us describe the measurement of the observable $X$ on $\hh_B=L^2(\R, Leb)$. We couple the system with the pointer of some measurement device, described by $\hh_B=L^2(\R, Leb)$. we call $\hh_B$ the pointer space (think of it as the needle of a weighting scale or a seismometer). We move the pointer depending on the value of $X$, which has the effect of applying a unitary operator $Z$ on $\hh_B\otimes \hh_B=L^2(\R^2, Leb_2)$ which is defined by 
\[
(Zf)(x,a)=f(x, a-x)\,.
\]

Then, we perform the measurement of the pointer: we measure $A=M_{a\mapsto a}$ on $\hh_B$. The result is a random variable $\tilde{A}$ and the state after the measurement is $\rho_{G|\cala}$ (where $\cala$ is the algebra generated by $A$). Note that the noise is described by the initial state of the pointer. For example, if the system is in the pure state $f \in L^2(\R, Leb)$ and the pointer in the pure state $g \in L^2(\R, Leb)$, the probability density of $\tilde{A}$ is 
\[
p(a)=\int_\R \abs{f(x)}^2\abs{g(a-x)}^2 dx=\abs{f}^2*\abs{g}^2(a)
\]
and for any $a \in \R$ the state $\rho{G|\cala}(a)$ is the pure state $\ket{f_a}\bra{f_a}$ where
\[
f_a(x)=\frac{f(x)g(a-x)}{p(a)}.
\]
This really corresponds to a classical noisy measurement : if $X$ is a random variable with density $\abs{f}^2$ and $B$ a random variable with density $\abs{g}^2$ then $p$ is the density of $X+B$ and $\abs{f_a}^2$ is the density of $X$ conditioned to $X+B=a$. Note however that this situation is truly quantum: if we do not perform the measurement, the density matrix of the system after the evolution is
\[
\rho_G'=\E(\tilde{\rho}_G)=\tra_B(Z\,(\rho_G\otimes\rho_B)\,Z^*)
\]
which is of kernel
\[
K_{\rho_G'}(x,y)=f(x)\overline{f(y)}\int_\R g(a-x)\overline{g(a-y)}da=f(x)\overline{f(y)}C_g(x-y)\,.
\]
where 
\[
C_g(z)=\int_\R g(a-z)\overline{g(a)}da.
\]
It is no more a pure state. 

A more general version of this process is the following:
 
\begin{defi}\label{def:indirect_meas}
Let $\hh_G$ be a Hilbert space and $\cala$ a commutative von Neumann algebra on $\hh_G$, with an isometry $\calg: L^2(\xx, \mu)\rightarrow \hh_G$ implementing an isomorphism $\cala\simeq L^\infty(\xx, \mu)$. Consider an auxiliary space $\hh_B=L^2(\yy, \nu)$. A \emph{pointer map} is some measurable function ${\psi : \xx\times \yy\rightarrow \xx}$ such that for all $x\in \xx$ the map $\psi(x,\bullet)$ is a measure-preserving bijection on $\yy$. 
The pointer unitary operator $Z_\psi$ on $\hh_G\otimes \hh_B$ corresponding to $\psi$ is the operator defined as $Z_\psi=\calg \tilde{Z}_\psi \calg^*$ where $\tilde{Z}_\psi$ is the unitary on $L^2(\xx\times \yy, \mu\times \nu)$ defined by
\[
( \tilde{Z}_\psi f)(x,y)=f(x,~\psi(x, y)~)\,.
\]
The indirect measurement corresponding to $\psi$ is the measurement of the algebra $L^\infty(\yy, \nu)$ on $\hh_B$, resulting in the random value $Y\in \yy$ of the pointer and the random state $\rho_{G|Y}\in \gs(\hh_B)$. 
\end{defi}

This is a little more restrictive than the processes considered by Belavkin \cite{Belavkin94}, in which the unitary operator $Z$ (written $S$ by Belavkin) is only assumed to commute with elements of $L^\infty(\xx, \mu)\otimes\{I_B\}$. This restrictive definition has the advantage of making it more explicit.

 This definition include the perfect measurement of a discrete observable $A$: take $\xx=\yy=sp(A)$ with $\mu$ the counting measure and fix an initial state $a_0 \in \yy$, choose $\rho_B=\ket{\delta_{a_0}}\bra{\delta_{a_0}}$ and any pointer function $\psi$ such that $\psi(a,a_0)=a$.

\subsection{ Repeated measurement process and the trajectories of the OQBM}

In this section we introduce repeated interactions and repeated measurement processes, and we show how the discrete OQBM can be seen as an extension of these measurement. We use this picture to show the convergence of the quantum trajectories of the discrete OQBM, thanks to a theorem of Pellegrini \cite{PellegriniDiffusive08}.

\subsubsection{The repeated measurement process}

The repeated measurement model relates to many experimental protocols, notably with the experiments of Haroche's team. It describes a process on discrete time, and we are interested in its continuous-time limit.

 We consider a Hilbert space $\hh_G$ describing a system of interest in the state $\rho_0\in \gs(\hh_G)$, and a space modeling a probe $\hh_p$ in the fixed pure state $\rho_p=\ket{0}\bra{0}$. In this article the probe space is always $\hh_p=\C^2$. Make it evolve according to some unitary operator $V$ on $\hh_G\otimes \hh_p$ and measure some observable $A\in \bb_{sa}(\hh_p)$. Then take a copy of $\hh_p$, also in the state $\rho_p=\ket{0}\bra{0}$, and repeat this procedure again and again. What we obtain is a stochastic process $(\rrho_n)_{n\in \N}$ where $\rrho_n\in \gs(\hh_G)$ is the state of the system after the $n$-th measurement, together with another process $(\Delta_n)_{n\in \N}$ where $\Delta_n\in \R$ is the result of the $n+1$-th measurement of $A$. Since the probe space $\hh_p$ is constantly renewed, $(\rrho_n, \Delta_n)_{n\in \N}$ is a Markov process. We can also note that for any $n$ the state $\rrho_n$ deterministically depends in the sequence $(D_k)_{k<n}$, since if $P_d$ is the spectral projection for the eigenvalue $d$ of $A$, we have
 \[
 \rrho_{n+1}=\frac{\tra_B\left(P_{\Delta_n} V (\rrho_n\otimes \rho_p) V^* P_{\Delta_n}\right)}{\tr{P_{\Delta_n} V (\rrho_n\otimes \rho_p) V^* P_{\Delta_n}}}~.
 \]

It is also interesting to study the evolution when the result of the measurement is discarded, that is, the evolution of $\rho_n=\E(\rrho_n)$. We have
\[
\rho_{n+1}=\tra_B\left(V (\rrho_n\otimes \rho_p) V^*\right)~
\]
The evolution of $\rho_n$ is called a quantum dynamical system, and its description as the interaction of the system with a bath is called a repeated interactions model \cite{AttalPautrat}.

 \subsubsection{The Belavkin diffusive Equation and the Lindblad Equation}\label{subsub:belavkinlimit}
 
We want to study the continuous time limit of this type of process. Thus, we will consider that each step of the process lasts a time $\tau>0$ and we make $\tau$ go to zero with suitable normalization. The case we consider is the following: 
\begin{enumerate}
\item We take $\hh_p=\C^2$ with $\rho_p=\ket{0}\bra{0}=\begin{pmatrix} 1 & 0 \\ 0 & 0 \end{pmatrix}$.\\
\item The unitary evolution $V_\tau$ on $\hh_G\otimes \hh_p$ is described as follows:  fix a self-adjoint bounded operator $H\in \bb(\hh_G)$ and a bounded operator $N\in \bb(\hh_G)$ and take
\begin{align}\label{eq:vtau}
V_\tau&=\exp\left(-i\tau H +\sqrt{\tau}\begin{pmatrix}
0 & N^* \\
-N & 0
\end{pmatrix}\right)\\
&=I+\sqrt{\tau}\begin{pmatrix}
0 & N^* \\
-N & 0
\end{pmatrix}
+\tau \left(-iH-\frac{1}{2}\begin{pmatrix}
N^* N & 0 \\
0 & N N^*
\end{pmatrix} \right)+O(\tau^{3/2})~.
\end{align}
\item We measure the observable $A=\begin{pmatrix} 0 & 1 \\ 1 & 0 \end{pmatrix}$.
\item The process of states obtained is written $(\rrho_{\tau,n})_{n\in \N}$, and the result of the $(n+1)$-th measurement is written $\Delta_{\tau, n}\in \set{-1, +1}$. We also define 
\[
W_{\tau, n}=\sqrt{\tau} \sum_{k=0}^{n-1} \Delta_{\tau,k}~.
\]
\end{enumerate}

The normalization in $\sqrt{\tau}$ to define $W_{\tau, n}$ corresponds to a diffusive limit in physics, where the time scale $\tau$ is proportional to the square of the space scale. {\bf In the rest of the article, we will write $\delta=\sqrt{\tau}$ the space scale.}

In this setup, the eigenvectors for the eigenvalues $\pm 1$ of $A$ are 
\begin{align}\label{eq:pm}
\ket{\pm}=\frac{1}{\sqrt{2}}\left(\ket{0}\pm\ket{1}\right)~,
\end{align}
 and we have
\[
\rrho_{\tau, n+1}=\frac{K_{\tau,\Delta_n}\rrho_{\tau, n} K_{\tau,\Delta_n}^*}{\tr{K_{\tau,\Delta_n}\rrho_{\tau, n} K_{\tau,\Delta_n}^*}}
\]
where 
\begin{align}\label{eq:BfromRM}
K_{\tau,\pm 1}&=\frac{1}{\sqrt{2}}\left(I\pm\delta N+\tau\left(-iH-\frac{1}{2} N^*N\right)\right)+O(\tau^{3/2})~.
\end{align}

The following theorem describes the limit in distribution of this process as $\tau\rightarrow 0$. It was proved by Pellegrini.

\begin{theo}[Theorem 8 of \cite{PellegriniDiffusive08}]\label{theo:pellegriniBelavkin}
Assume that $\hh_G$ is finite-dimensional. Fix some $T>0$. Then the process $(\rrho_{\tau, \ent{t/\tau}}, W_{\tau, \ent{t/\tau}})_{0\leq t \leq T}$ described above converges in distribution as $\tau\rightarrow 0$ (in the space of bounded functions with the uniform norm) to a process $(\rrho_t, W_t)_{0\leq t \leq T}$ satisfying the following stochastic equation (in the It\^o sense):
\begin{align}\label{eq:belavkin}
\left\{\begin{array}{ll}
d\rrho_t&=\call(\rrho_t)dt+(N\rrho_t+\rrho_t N^*-\rrho_t\Tau(\rrho_t))dB_t \\
dW_t&=\Tau(\rrho_t)dt+dB_t\\
\end{array}\right.
\end{align}
where $(B_t)$ is a standard Wiener process, $\call$ is the super-operator defined by
\begin{align}\label{eq:call}
\call(\rho)=-i[H, \rho]+N\rho N^*-\frac{1}{2}\left(N^*N\rho+\rho N^*N\right)~
\end{align}
and
\begin{align*}
\Tau(\rho)&=\tr{(N+N^*)\rho}~.
\end{align*}
\end{theo}

This theorem was proved with methods of classical stochastic process, notably the Kurtz-Protter's theorem. Importantly, the proof is still valid with $K_{\tau, \pm}$ replaced by $K_{\tau, \pm}+o(\tau)$ with some rest $o(\tau)$ uniformly small.  
\vspace{0.5cm}

If we discard the probes before measuring it, the state of the system is the deterministic density matrix 
\[
\rho_{G,\tau, n}=\E(\rrho_{\tau, n})~.
\]
It follows a quantum dynamical semigroup, with $\rho_{G,\tau, n+1}=\Lambda{G,\tau}(\rho_{G,\tau, n})$ where 
\begin{align}\label{eq:defi_lambdaG}
\Lambda_{G,\tau}(\rho)&=K_{\tau,+1}\rho K_{\tau,+1}^*+K_{\tau,-1}\rho K_{\tau,-1}^*\\
&=\rho+\tau\call(\rho)+O(\tau^{3/2})
\end{align}
where $\call$ is defined in Equation \eqref{eq:call}.
Thus $\rho_{G,\tau, \ent{t/\tau}}$ converges to some limit $\rho_{G,t}$ satisfying the so-called Lindblad Equation
\[
\frac{d}{dt}\rho_{G,t}=\call(\rho_{G,t})~.
\]
The family of super-operators $\Lambda^t=e^{t\call}$ is called a Lindblad semigroup. Note that $\rho_t=\E(\rrho_t)$, which can be seen both by the above convergence or by using the fact that the term in $dB_t$ in Equation \eqref{eq:belavkin} is of expectancy zero.

\subsubsection{A dilation of the discrete OQBM}\label{subsub:discreteOQBM}

Let us consider a gyroscope space $\hh_G$ and the position space $\hh_{\tau, z}=l^2(\delta \Z)$ (for $\delta=\sqrt{\tau}>0$), and fix some bounded operators $N$ and $H$ on $\hh_G$ with $H$ self-adjoint, and consider the operators $B_{\tau, +}$ and $B_{\tau, -}$ defined as in the introduction
(with $M=0$ since the effects of $M$ are negligible). We consider the two channels defining the OQBM: the one which corresponds to the OQW definition:
\[
\Lambda_\tau(\rho)=\sum_{x\in \delta\Z} \left(B_{\tau, -} \otimes \ket{x-\delta}\bra{x} \right)\rho\left(B_{\tau,-}^*\otimes \ket{x}\bra{x-\delta}\right) +\left(B_{\tau, +} \otimes \ket{x+\delta}\bra{x} \right)\rho\left(B_{\tau, +}^*\otimes \ket{x}\bra{x+\delta}\right)
\]
and the one with only two Krauss operators:
\begin{align*}
\tilde{\Lambda}_\tau(\rho)=\left(B_{\tau, -}\otimes D_\tau^*\right) \rho \left(B_{\tau, -}^*\otimes D_\tau\right)+\left(B_{\tau, +}\otimes D_\tau\right) \rho \left(B_{\tau,+}^*\otimes D_\tau^*\right)
\end{align*}
where $D_\tau$ is the right translation of distance $\delta$ on $\hh_{\tau, z}$. It is easyly checked that $\Lambda_\tau^*$ and $\tilde{\Lambda}_\tau^*$ coincide on the algebra $\bb(\hh_G)\otimes L^\infty(\delta \Z)$, and we concentrate on the study of $\tilde{\Lambda}_\tau$ from now on. To make the link with the repeated measurement process and the Belavkin Equation, we define one Stinespring dilation of $\tilde{\Lambda}_\tau$.

\begin{lem}\label{lem:dilation}
We have 
\begin{align}\label{eq:dilationOQBM}
\tilde{\Lambda}_\tau(\rho)=\tra_{\hh_p}\left(R_\tau V_\tau \left(\rho\otimes\ket{0}\bra{0}\right) V_\tau^*R_\tau^* \right)+O(\tau \sqrt{\tau})
\end{align}
where $\hh_p=\C^2$ and $V_\tau$ is the unitary operator on $\hh_G\otimes \hh_p$ defined by Equation \eqref{eq:vtau} and $R_\tau$ is the operator on $\hh_{\tau, z}\otimes \hh_p$ defined by
\begin{align}\label{eq:rtau}
R_\tau=D_\tau \otimes \ket{+}\bra{+}+D_\tau^*\otimes \ket{-}\bra{-}
\end{align}
where $\ket{+}=\frac{1}{\sqrt{2}}\left(\ket{0}+\ket{1}\right)$ and $\ket{-}=\frac{1}{\sqrt{2}}\left(\ket{0}-\ket{1}\right)$.
\end{lem}
This lemma proved by a straighforward computation, and a consequence of the equality $L_{\tau,\pm 1}=B_{\tau, \pm}$ where $L_{\tau, \pm 1}$ is defined in the repeated measurement procedure \eqref{eq:BfromRM}.

The notation $O(\tau \sqrt{\tau})$ is meant uniform in $\rho$, in the sense that there exists a constant $C$ such that for all $\tau >0$ small enough and all $\rho\in \gs(\hh_G\otimes \hh_{\tau,z})$ we have $\norm{O(\tau \sqrt{\tau})}\leq \tau \sqrt{\tau} $.

As a consequence, we see that $\tilde{\Lambda}_\tau$ is an extension of the quantum dynamics $\Lambda_{G, \tau}$ on $\hh_G$ generated by repeated interactions: 
\begin{cor}
For any initial state $\rho\in \gs(\hh_G\otimes \hh_{\tau,z}$ we have 
\[
\tra_{\hh_{\tau,z}}\left(\tilde{\Lambda}_\tau(\rho) \right)=\Lambda_{G, \tau}\left(\tra_{\hh_{\tau,z}}(\rho)\right)
\]
(where $\Lambda_{G, \tau}$ is defined by the repeated measurement process, Equation \eqref{eq:defi_lambdaG}).
\end{cor}

\begin{proof} This can be proved by direct computation, but it is also a consequence of Lemma \ref{lem:dilation}. Indeed, $R_\tau$ does not act on $\hh_G$ so it commutes with any operator $E$ on $\hh_G$ and we have
\begin{align*}
\tr{\tilde{\Lambda}_\tau(\rho) E}&= \tr{ R_\tau V_\tau (\rho\otimes \ket{0}\bra{0}) V_\tau^* R_\tau^* E} \\
&=\tr{V_\tau (\rho \otimes \ket{0}\bra{0} )V_\tau^* E}  \\
&=\tr{\Lambda_{G,\tau}(\rho) E}~.
\end{align*}
which proves the corollary.
\end{proof}

\subsubsection{Convergence of the quantum trajectories}

The following theorem is a direct consequence of Pellegrini's theorem \ref{theo:pellegriniBelavkin} and of the picture of the discrete OQBM as an extension of the repeated interactions process:

\begin{prop}\label{prop:trajectoryOQBM}
Let $(\rrho_0, X_0)\in \gs(\hh_G)\times \R)$ be a random variable. For any $\tau>0$ let us consider the process $(\rrho_{\tau,n}, X_{\tau, n})_{n\in \N}$ describing the quantum trajectories of the OQBM (defined in \eqref{eq:defi_traj_oqw}) with inital state $(\rrho_0, \delta\ent{X_0/\delta})*$. Then for any $T>0$ the family of processes $(\rrho_{\tau, \ent{t/\tau}}, X_{\tau, \ent{t/\tau}})_{t\in [0, T]}$ converges in distribution as $\tau\rightarrow 0$ to a process $(\rrho_t, X_t)_{t\in [0, t]}$ satisfying the following differential equation:
\begin{align}\label{eq:belavkinOQBM2}
\left\{
\begin{array}{ll}
d\rrho_t&=\call(\rrho_t)dt+\left(N\rrho_t+\rrho_t N^*-\rrho_t\Tau(\rho_t)\right)dB_t\\
dX_t&=\Tau(\rrho_t)dt+dB_t
\end{array}
\right.
\end{align}
where $B_t$ is a Wiener process.
\end{prop}

\begin{proof}
We have $B_{\tau,\pm}=K_{\tau, \pm}+O(\tau \sqrt{\tau})$; we can ignore the rest $O(\tau \sqrt{\tau})$, since in the proof of Theorem \ref{theo:pellegriniBelavkin} (as exposed in \cite{PellegriniDiffusive08}) does not depends on the terms which are of order $o(\tau)$. Thus, the process $(\rrho_{\tau, \ent{t/\tau}}, X_{\tau, \ent{t/\tau}}-X_0)_{t\in [0, T]}$ has the same limit in distribution as the process $(\rrho_{\tau, \ent{t/\tau}}, W_{\tau, \ent{t/\tau}})_{t\in [0, T]}$, which satisfies Equation \eqref{eq:belavkinOQBM2}.
\end{proof}

The fact that $B_{\tau, \pm}=K_{\tau, \pm}+O(\tau \sqrt{\tau})$ can be directly computed, but it is also a consequence of Lemma \ref{lem:dilation}: the unitary $R_\tau$ converts the measurement of the observable $A$ into the measurement of the increasing of the position. Thus, the quantum trajectories of the OQBM are nothing more than the trajectories of the Belavkin Equation; the OQBM is truly different from the quantum dynamics arising from $V_\tau$ when we consider the position $X_t$ of the particle as a quantum observable, that is in the Lindbladian and the QSDE versions of the OQBM.

\subsection{Quantum Stochastic Calculus for the Open Quantum Brownian Motion}\label{sub:qsc}

  A fully quantum view on the OQBM which encompass the quantum correlations between the events at different times is obtained with the Quantum Stochastic Calculus on the Fock space. We will briefly introduce the Fock space and quantum stochastic calculus, by approaching it by the repeated interactions process.

\subsubsection{Repeated interaction process and the Toy Fock space}\label{subsub:toyfock}

In the definition of the repeated interactions process, a new probe space $\hh_p$ is introduced at every iteration. The so called Toy Fock space is the Hilbert space $T\Phi$ obtained when considering all these probe spaces at once. Formally, $T\Phi=\bigotimes_{n\in \N^*} \hh_p$. More concretely, it is the Hilbert space which generated by the vectors $\bigotimes_{n\in \N^*} e_n$ where the vectors $e_n$ are unit vectors of $\hh_p$ which are all equal to $\ket{0}$ except for a finite number of indexes. It has a distinguished unit vector $\ket{\Omega}=\bigotimes_{n\in \N^*} \ket{0}$, and for each $n\in \N^*$ it can be naturally decomposed as 
\[
T\Phi\simeq \hh_p^{\otimes n} \otimes T\Phi~.
\]
This identification is implicit in the following.

The evolution correspondint to the $n$-th interaction is described by the operator $V_\tau$ acting on the $n$-th copy of $\hh_p$, i.e. the operator  $V_{\tau, n}=\Un_{\hh_p^{\otimes (n-1)}}\otimes V_\tau \otimes \Un_{T\Phi} $, and the evolution from time zero to time $n$ is represented by the unitary
\[
U_{\tau, n}=V_{\tau, n} V_{\tau, n-1}\cdots V_{\tau, 1}~.
\]

For each $n\in \N^*$, the space $T\Phi_d$ contains a copy of $\hh_p$ given by the isometry
\[
\begin{array}{llll}
\calg_n :& \hh_p & \longrightarrow  & T\Phi_d\\
& v&\mapsto & \left(\overset{n-1}{\underset{k=1}{\bigotimes}} \ket{0}\right)\otimes v\otimes \left(\overset{+\infty}{\underset{k=n+1}{\bigotimes}}\ket{0}\right)
\end{array}
\]

We can obtain the random state $\rrho_n$ by performing the simultaneous measurement of all the observables $A_k=\calg_n A \calg_n^*$ when in the total state
\[
\rho_{tot, \tau, n}=U_{\tau, n} \left(\rrho_{\tau,0}\otimes \ket{\Omega}\bra{\Omega} \right) U_{\tau, n}^*~.
\]

The position of the particle is then 
\[
X_{\tau, n}=X_{\tau, 0}+\delta \sum_{k=1}^n \tilde{A}_k
\]
where  $\tilde{A}_k=\pm 1$ is the result of the measurement of $A_k$.

\subsubsection{The Fock space} \label{subsub:fock}

Before studying the convergence of $T\Phi$ as $\tau\rightarrow 0$, let us describe its limit, the Fock space $\Phi=\bigotimes_{t\in \R_+} \hh_p$. This space and its interpretation as an infinite tensor product is well known, see Parthasarathy's book \cite{Parthasarathy1992} for example, or Attal's lecture in the second book of \cite{OQS}, and we refer to these lectures for a more complete introduction to the Fock space. Let us briefly recall two of its descriptions. Here, we only treat the case where $\hh_p=\C^2$, but the case where $\hh_p=\C^n$ or even $\hh_p $ is infinite-dimensional are similar.

{\bf The Guichardet interpretation:} Let us consider the set $\pp$ of increasing sequences of $\R_+$ of finite length (including the empty sequence $(\emptyset)$). We have $\pp=\cup_{n\in \N} \pp_n$ where $\pp_n\subset (\R_+)^n$ is the set of increasing sequence of length $n$. This set inherits the Lebesgue measure on $(\R_+)^n$ (and $\pp_0=\set{(\emptyset)}$ has the Dirac measure), so we can endow $\pp$ with the sum of these measure, which we write $\lambda$. The Fock space in the Guichardet interpretation is $\Phi_G=L^2(\pp, \lambda)$.

It can be interpreted as an infinite tensor product. Indeed, if we write $\pp_{[s,t]}$ the space of finite sequences in $[s,t]$ and $\Phi_{G,[s,t]}=L^2(\pp_{[s,t]}, \lambda)$, we have $\Phi_{G,[s,t]}\otimes \Phi_{G,[t,u]}=\Phi_{G,[s,u]}$. There is a distinguished vector $\ket{\Omega}=\Un_{\pp_0}$. We identify $\Phi_{G, [s,t]}$ to the subspace $\set{\ket{\Omega_{[0, s]}}}\otimes\Phi_{G, [s, t]}\otimes \set{\ket{\Omega_{[t, +\infty)}}}$ of $\Phi_G$.

{\bf The probabilistic interpretation from the Brownian motion:} This interpretation has been introduced by Attal and Meyer \cite{AttalMeyer}. See \cite{AttalOQS} for more details. We consider the Wiener space $(\calw, \ff)$ of continuous functions from $\R_+$ to $\R$ with the Wiener measure $\mu$ corresponding to the Browian motion. We then take $\Phi_W=L^2(\calw, \mu)$ the space of $L^2$ random variables on $(\calw, \mu)$. There is a distinguished vector $\ket{\Omega}=\Un$ (the constant random variable equal to 1). If $\calw([s,t])$ is the space of functions from $[s,t]$ to $\R$, we can define $\Phi_{W, [s,t]}=L^2(\calw([s,t]), \mu)$, and we have
$\Phi_{W,[s,t]}\otimes \Phi_{W,[t,u]}=\Phi_{W,[s,u]}$. 

These two interpretation are equivalent: we can construct an unitary operator $\calg_{G, W} : \Phi_G \rightarrow \Phi_W$ such that $\calg_{G, W} \Phi_{G, [s,t]}=\Phi_{W, [s,t]}$ and $\calg\ket{\Omega}=\ket{\Omega}$. To describe it, let us write $(W_t)_{t\in \R_+}$ the Brownian motion and $dW_t$ the It\^o differential. For any function $f\in L^2(\pp_n, \lambda)$, the random variable $X=\calg_{G, W}f$  is defined as the successive It\^o integrals
\[
\calg_{G, W}f=X=\int_{0<t_1<t_2<\cdots<t_n<\infty} f(t_1, \cdots, t_n)dW_{t_1} dW_{t_2}... dW_{t_n}
\]

(and if $n=0$ then $\calg_{G, W}f$ is the deterministic variable equal to $f(\emptyset)$).

By the It\^o isometry formula, we have
\[
\norm{X}^2=\E(\abs{X}^2)=\int_{0<t_1<t_2<\cdots<t_n<\infty} \abs{f(t_1, \cdots, t_n)}^2dt_1\cdots dt_n=\norm{f}^2~.
\]
so $\calg_{G, W}$ is an isometry, and the chaotic representation property ensure that it is surjective (see \cite{AttalOQS}). 

From now on, we will write $\Phi$ the Fock space, and either the Guichardet or the probabilistic interpretation depending on the context. There exists many more probabilistic interpretations, one for each normal martingale. We concentrate on the Brownian interpretation in this article.
\vspace{0.5cm}

To complete this picture, we need to approximate the Toy Fock space by the Fock space. This was done by Attal \cite{AttalToy} and developed by Attal and Pautrat \cite{AttalPautrat}. Let us first design an isometry of $T\Phi$ into $\Phi$. The idea is the following: for each $\tau$, we have
\[
\Phi=\bigotimes_{n\in \N} \Phi_{[\tau n~,~\tau(n+1)]}
\]
(where the infinite tensor product is taken with respect to $\ket{\Omega_{[\tau n~,~\tau(n+1)]}}$ as in the construction of the toy Fock space). 
Thus, it is sufficient to define an isometry from $\hh_p=\C^2$ to $\Phi_{[\tau n~,~\tau(n+1)]}=\Phi_{[0, \tau]}$ and to extend it by tensor product to $T\Phi=\otimes_{n\in \N^*}$. We choose the isometry
\[
\begin{array}{llcl}
\calg_{n,\tau}: & \hh_p & \longrightarrow & \Phi_{[\tau n~,~\tau(n+1)]} \\
& \ket{0} & \mapsto & \ket{\Omega_{[\tau n~,~\tau(n+1)]}}\\
&\ket{1} & \mapsto &  \frac{1}{\sqrt{\tau}}\left(W_{\tau(n+1)}-W_{\tau n}\right)~.
\end{array}
\]
which tensorise to $\calg_{\tau}=\otimes_{n\in \N} \calg_{n,\tau}: T\Phi\rightarrow \Phi$.

Let us write $P_\tau=\calg_{\tau} \calg_{\tau}^*$ the projection on the image of $\calg_{\tau}$, and $T_\tau \Phi$ this image. 
Then $P_\tau$ strongly converge to the identity on $\Phi$ as $\tau\rightarrow 0$. 
In this sense, the Toy Fock space approximate the Fock space, but this is not sufficient; we also need some more precise convergence on operators in $\bb(\Phi)$. But first, we need to study the operators in the Fock space.

\subsubsection{Quantum Stochastic Calculus on the Fock space}

The quantum stochastic calculus is thoroughly described in Parthasarathy \cite{Parthasarathy1992} and in \cite{AttalOQS}, \cite{AttalPautrat}. We give it a very short introduction geared for this article.
 
The operators on $\hh_p$ are all linear combinations of the four operators $\ket{j}\bra{i}$ for $i,j\in \{0,1\}$. In the toy Fock space, they translate as the operators
\begin{align*}
a^i_j(n)=\calg_n(\ket{j}\bra{i}) \calg_n^*~.
\end{align*}
Thus, the algebra $\bb(T\Phi)$ is generated by the operators $a^i_j(n)$ for $n\in \N^*$ and $i,j\in \{0,1\}$. Unter suitable renormalization, they converge as $\tau\rightarrow 0$. Using the isometry $\calg_\tau$ in the Fock space, we define the operator
\[
a^i_j(\tau, k, l)=\calg_\tau \prod_{n=k}^l a^i_j( n) \calg_\tau
\]
 then there exists closed operators $a^i_j(t)$ on $\Phi$ such that there is strong convergence
\[
\tau^{\eps_{j,i}} a^i_j(\tau,0,\ent{t/\tau}) \underset{\tau\rightarrow 0}{\longrightarrow} a^i_j(t)
\]
where 
\[
\tau^{\eps_{j,i}}=\left\{\begin{array}{cl}
\tau & \text{if $i=j=0$}\\
\sqrt{\tau} & \text{if $(i,j)=(0,1)$ or $(i,j)=(1,0)$}\\
1 & \text{if $i=j=1$}
\end{array}\right.~.
\]
The operator $a^0_0(t)$ is just the multiplication by $t$, while $a^1_0(t)^*=a^0_1(t)$ and $a^1_1(t)$ is self-adjoint (they are respectively the creation, annihilation and number operator on $\Phi_{t]}$. We write $a^i_j([s,t])=a^i_j(t)-a^i_j(s)$; we have
\[
a^i_j(\tau,n)= \tau^{-\eps{j,i}}~P_\tau a^i_j([\tau(n+1),~\tau n]) P_\tau ~.
\]
 See \cite{AttalToy} or \cite{AttalPautrat} for more details on these operators. We will now explain how to integrate with respect to theses operators, in a way parallel to the It\^o Stochastic integration. First, we need to define the set of coherent vectors. For any function $u\in L^2(\R)\cap L^\infty(\R)$, we define the coherent vector $\eps(u)$ in the Guichardet interpretation by
\[
\eps(u)(t_1, \cdots, t_n)=u(t_1)u(t_2)\cdots u(t_n)~
\]
(the empty product being considered to be $1$).
In he probabilistic interpretation, it corresponds to exponential martingale : writing $Y_t=\eps(u\Un_{[0,t)})$ it verifies the (classical) SDE
\[
dY_t=u(t)Y_t dW_t
\]
 Thus, writing $H_\infty=\int_0^{+\infty} u(s) dW_s$ and $[H]_\infty=\int_{0}^{+\infty} \abs{u(s)}^2 ds$   we have 
\[
\eps(u)=\exp\left(H_\infty-\frac{1}{2}[H]_\infty\right)~.
\]
We have $\norm{u}^2=e^{\norm{u}^2_{L^2}}$. Hence, $\eps$ is continuous; it is clearly not linear.

An important property is that if $\mm\subset L^2(\R)\cap L^\infty(\R)$ is a dense subspace of $L^2(\R)$, then the vector space $Vect(\eps(\mm))$ is dense in $\Phi$. Thus, it is often sufficient to define an object on coherent vectors to fix it.

Now, the objects that we can integrate are the adapted process of operators. We give here a restrictive definition taken from Parthasarathy \cite{Parthasarathy1992}. A more general definition was produced by Attal and Lindslay \cite{AttalLindsay} , but it is not needed here.

\begin{defi}[Adapted process of operators]\label{defi:adapted}
A dense subspace $\mm\subset L^2(\R)$ is called adapted if for any $0\leq s\leq t\leq \infty$, the space $\mm([s,t]):=\{f\in \mm|f=\Un_{[s,t]}f\}$ is dense in $L^2([s,t])$.

Consider some Hilbert space $\hh_S$. 
A family of (possibly unbounded) operators $(H_t)_{t\in \R_+}$ on $\hh_S\otimes\Phi$ is called adapted 
if there exists a dense subspace $\dd$ and an adapted subspace $\mm\subset L^2(\R)$ such that for all $t$ the domain of $H_t$ contains $\dd\otimes_{alg}~ \eps(\mm)$,
and there is an operator $\tilde{H_t}$ on $\hh_S\otimes \Phi$ with domain $\dd\otimes_{alg}~\eps(\mm([0,t])$ 
such that $H_t=\tilde{H_t}\otimes I_{\Phi_{[t,+\infty)}}$ on $\dd\otimes_{alg}~\eps(\mm)$.
\end{defi}

Now, for an adapted process of operators $(H_t)_{t\in \R_+}$, we want to define the operator
\[
\int_0^t H_s da^i_j(s)
\]
which would correspond to the limit of 
\begin{align}\label{eq:int_riem}
\frac{1}{\tau} \sum_{k=0}^{\ent{t/\tau}} H_{k\tau} \left(a^i_j(\tau(k+1))-a^i_j(\tau k)\right)~.
\end{align}
Note that $a^i_j(\tau(k+1))-a^i_j(\tau k)$ only acts on $\Phi_{[\tau k, \tau(k+1)]}$ so it commutes with $H_{k\tau}$, and the order of the operators in the above formula is not important. The concrete way we define the integral is the following:

\begin{defi}
Let $(H_t)_{t\in I}$ be an adapted process of operators on $\hh_S\otimes \Phi$, with domain containing $\dd\otimes_{alg}~ \eps(\mm)$ where $\mm$ is adapted and $\dd$ is dense. Let $T$ be an operator on $\hh_S\otimes \Phi$. We say that the formula
\[
T=\int_0^t H_t da^i_j(t)
\]
is true on  $\dd\otimes_{alg}~ \eps(\mm)$ if for any $a,b\in \dd$ and $u,v\in \mm$ the following formula is meaningful and true:
\begin{align}\label{eq:int_coherent}
\scal{a\otimes \eps(u)~,~T_t b\otimes\eps(v)}=\int_0^t u_j(s)u_i(s)~\scal{a\otimes \eps(u),H_s b\otimes \eps(v)} ds
\end{align}
where $u_i(s)=1$ if $i=0$ and $u_i(s)=u(s)$ if $i=1$, and by \enquote{meaningful} we mean that the integral is absolutely convergent. 
\end{defi}

If $T_t=\int_0^t H_s da^i_j(s)$ for all $t$ we will write $dT_t=H_t da^i_j(t)$. A more general formula exists to compute $T f$ for some vector $f$, see \cite{AttalOQS}. Note that the existence of an operator $\int_0^t H_t da^i_j(t)$ is not guarantied. If $H_t$ is bounded locally uniformly in $t$, it is at least possible to define $\int_0^t H_t dt$ on the space generated by ${\hh_S\otimes \dd_B}$, where $\dd_B$ is the vector space generated by $\eps(L^2(\R)\otimes L^\infty(\R))$. The obtained operator may still be unbounded. 
\vspace{0.5cm}

 It is easy to check that in the case where $H_t$ is constant on the intervals $t\in [\tau k, \tau(k+1)]$ this formula corresponds to the Riemann sum \eqref{eq:int_riem}. In particular, 
 \[
 a^i_j(t)=\int_0^t da^i_j(s)~.
 \] 
 The case of $a^0_0(t)=t$ is simple, the integral being just the integral with respect to $dt$ in the Banach space $\bb(\hh_S)$.
 
 The case of $a^1_0(t)$ and $a^0_1(t)$ is more subtle, and it actually generalize the It\^o integral, as shown by the following proposition.
 
 \begin{prop}
 Let $(f_t)_{t\in \R_+}$ be a process of random variables in $L^\infty(\calw, \mu)$, adapted in the sense of It\^o, and such that $\int_0^t \E\big(\abs{f_s}^2\big)ds<\infty$. Let 
 \[
 g=\int_0^t f_s dW_s~.
 \]
 Consider the operators $H_s=M_{f_s}$ and $T=M_{g}$ on multiplication by $f_s$ on $\Phi$. Then we have
 \[
 T=\int_0^t H_s (da^1_0(t)+da^0_1(t))
 \]
 on the domain $\eps(L^2(\R))$. Thus, in terms of operators, we can write $dW_t=da^1_0(t)+da^0_1(t)$. 
 \end{prop}

By the predicable representation property (see \cite{AttalOQS}), this implies that the commutative von Neumann algebra $\cala([0,t])=L^\infty(\calw([0,t]), \mu)$ is generated by the operators $a^1_0(s)+a^0_1(s)$ for $s\leq t$. Note that the observable we measure in the definition of the OQBM is $A=\ket{0}\bra{1}+\ket{1}\bra{0}$, so the observable $A(\tau, n)=\calg_\tau A(n) \calg_{\tau}^*$ is 
\[
A(\tau, n)=\frac{1}{\sqrt{\tau}}P_\tau\tau(a^1_0([\tau n~,~\tau(n+1)])+a^0_1([\tau n~,~\tau(n+1)]))P_\tau~.
\]
Thus, the algebra generated by the $A(\tau,k)$ for $k\leq n$ is $P_z L^\infty(\calw([0,t],\mu) P_z$, which is the reason why the Brownian representation of $\Phi$ is adapted to the study of the OQBM.
\vspace{0.5cm}

The product of two quantum stochastic integrals is itself a quantum stochastic integral under some regularity conditions.

\begin{prop}[Quantum It\^o product formula] \label{prop:quantum_ito}

Let $(A_t)_{t\in \R_+}$ and $(B_t)_{t\in \R_+}$ be two adapted processes of operators, with domains containing the dense adapted domains $\dd\otimes_{alg}~ \eps(\mm)$ in $\hh_S\otimes \Phi$. Assume that $(A_t^*)_{t\in \R_+}$ is also an adapted process with domain containing $\dd\otimes_{alg}~ \eps(\mm)$ and that the following integrals are well defined, on $\dd\otimes_{alg}~ \eps(\mm)$:
\begin{align*}
T_t&=\int_0^t A_s da^i_j(s) & S_t&=\int_0^t A_s^* da^j_i(s) \\
U_t&=\int_0^t B_s da^k_l(s)~. 
\end{align*}

Moreover, assume that for all $s$ the operators  $A_s U_s$, $T_s B_s$ and $A_s B_s$ are defined on a domain containing $\dd\otimes_{alg}~\eps(\mm)$ and that the following integrals are well defined on this domain:
\begin{align*}
\int_0^t A_s U_s da^i_j(s) & & \int_0^t T_s B_s da^k_l(s)  && \int_0^t \delta_{i=l} \delta_{l\neq 0}A_s B_s da^k_j(s)~.
\end{align*}
Then the following formula is satisfied on $\dd_B\otimes \eps(\mm_B)$:
\[
T_t U_t=\int_0^t A_s U_s da^i_j(s)+T_s B_s da^k_l(s)+ \delta_{i=l} \delta_{l\neq 0}A_s B_s da^k_j(s)~.
\]
\end{prop}
This proposition was proved by Hudson and Parthasarathy, see Proposition 25.26 of Parthasarathy's book \cite{Parthasarathy1992}.

Writing $da^i_j(s)~da^k_l(s)=\delta_{i,j}\delta_{k\neq 0} ~da^i_l(s)$, this formula can be used as
\[
d(T_t U_t)=T_tdU_t+(dT_t)U_t+(dT_t)(dU_t)~.
\]
Note that in particular, if $A_t=B_t=a^1_0(t)+a^0_1(t)$ we have
\[
d(A^2(t))=2 A(t) dA(t)+ dt
\]
which is actually the formula $d(W_t^2)=2W_t dW_t+dt$ for the Brownian motion.

We are now ready to present the theorem of convergence of the repeated interactions of Attal and Pautrat.

\subsubsection{Hudson-Parthasarathy Equations and Attal-Pautrat convergence}

The Attal-Pautrat limit \cite{AttalPautrat} was devised in the context of repeated interactions processes. The idea is to show that $\calg_\tau U_{\ent{t/\tau}}\calg_\tau^*$ converge to some limit $U_t$ as $\tau$ goes to $0$, which satisfies a quantum stochastic differential equation.  We only present the case which is needed here.

First, we need to describe what will be the limit. It is a family of unitary following the so called quantum Langevin Equations (or Hudson-Parthasarathy Equations).

\begin{theo}\label{theo:HP}
Let $H$ and $N$ be two bounded operators on $\hh_G$ with $H$ self-adjoint. Write 
\[
G=-iH -\frac{1}{2} N^* N~.
\]
Then there exists an adapted process of unitary operators $U_t$ on $\hh_G\otimes \Phi$ which satisfies the following quantum stochastic equation on $\hh_G\otimes_{alg}~ \eps(L^2(\R))$:
\begin{align}\label{eq:HP}
dU_t=\big(G dt+N da^1_0(t)-N^*da^0_1(t)\big) U_t~.
\end{align}
The adjoint operator $U_t^*$ satisfies the adjoint equation. With the condition $U_0=I$, it is unique.
\end{theo}

This theorem is proved in \cite{Parthasarathy1992}; the idea is to make Picard iterations on Equation \eqref{eq:HP} starting from $U^0_t=I$, applying Formula \eqref{eq:int_coherent} to show that at each iteration the obtained operators are still unitary.
\vspace{0.5cm}

Attal and Pautrat proved the following theorem (in a more general setup). 
\begin{theo}\label{theo:AP}
Let $(U_{\tau, n})_{n\in \N}$ be a family of operators on $\hh_G\otimes T\Phi$ defined as in Paragraph \ref{subsub:toyfock}, and write $u_{\tau, n}=\calg_\tau U_{\tau, n} \calg_{\tau}^*$ the isometry on $\hh_G\otimes \Phi$ corresponding to $U_{\tau, n}$. 
Then for any $t\geq 0$ the operator $u_{\tau, \ent{t/\tau}}$ converges strongly to the unitary operator $U_t$ solution of the Hudson-Parthasarathy Equation of Theorem \ref{theo:HP}.
\end{theo}

This theorem is proved in \cite{AttalPautrat} in a more general context where there may be some term in $da^1_1(t)$ in the equation and the space $\hh_p$ is of arbitrary dimension).

\subsubsection{Convergence to the continous OQBM}

We are now ready to prove the convergence of the discrete OQBM. We consider the unitary operator $R_\tau V_\tau$  of the discrete OQBM built in Paragraph \ref{subsub:discreteOQBM}. The isometry $\calg_n$ converts it in a unitary on $\hh_G\otimes \hh_{\tau, z}\otimes \Phi$. The goal is to show the convergence to a unitary on $\hh_G\otimes L^2(\R)\otimes \Phi$, so we need to see $\hh_{\tau, z}$ as a subspace of $\hh_z=L^2(\R)$. For each $\tau$ we define an isometry of $\hh_{\tau,z}$ into a subspace of $L^2(\R)$. 
\[
\begin{array}{llcl}
\calg_{\delta\Z, \R}: & \hh_{\tau,z}& \longrightarrow & L^2(\R) \\
& \ket{x} &\mapsto & \frac{1}{\delta}~\Un_{[ x~,~x+\delta)}
\end{array}
\]
The image of this isometry is the space of functions which are constant on each interval $[ x~,~x+\delta)$, which we identify with $\hh_{\tau, z}$ in the following of the article, and we write $\hh_z=L^2(\R)$.  We define $P_{\delta\Z}=\calg_{\delta\Z,\R}\calg_{\delta\Z, \R}^*$ the orthogonal projection on this space. By the Lebesgue differentiation Theorem, it strongly converge to the identity as $\delta\rightarrow 0$. In this sense, the space $\hh_{\tau, z}$ converges to $L^2(\R)$ as $\tau\rightarrow 0$.

Moreover, the translation operator $D_\tau\in \bb(\hh_z)$ is transformed into 
\[
\calg_{\delta \Z,\R} D_\tau\calg_{\delta\Z, \R}^*=P_{\delta \Z} e^{-\delta\partial_x}
\]
since $e^{-\delta\partial_x}=e^{-i P}$ is the translation operator on $L^2(\R)$. 

Let us write 
\[
l_{\tau, n}=\calg_{\tau}\calg_n R_{\tau}V_\tau \calg_n^* \calg_\tau^*
\]
 and define the OQBM isometry 
 \[
 \gu_{\tau, n}=l_{\tau,n}l_{\tau, n-1}\cdots l_{\tau, 1}~.
 \]  
 
We have the following convergence theorem.

\begin{theo}\label{theo:HPOQBM}
For each $t\geq 0$ the operator $gu_{\tau, \ent{t/\tau}}$ converge strongly to some unitary operator $\gu_t$ solution of the equation
\begin{align}\label{eq:HPOQBM2}
d\gu_t=\Big((-iH-\frac{1}{2}N^*N+\frac{1}{2}\partial_x^2-\partial_x N)dt+(N-\partial_x)da^0_1(t)+(-N^*-\partial_x)da^1_0(t) \Big) \gu_t~
\end{align}
on the set $\hh_G\otimes_{alg}~ H^2(\R)\otimes_{alg}~ \eps(L^2(\R))$. 
\end{theo}

\begin{rem}
\begin{enumerate}
\item This theorem can probably be generalized to cases where $N$ and $H$ depends on the position $x$, but this would require to extend non-trivially the theorem of Attal and Pautrat, the issue of the non-boundedness of $\partial_x$ being harder to bypass when $N$ and $\partial_x$ are not commuting.
\item Equation \eqref{eq:HPOQBM} is a Hudson-Parthasarathy Equation of the form of Theorem \ref{theo:HP}, with $N$ replaced by $\tilde{N}=N-\partial_x$ and $H$ replaced by $\tilde{H}=H-\frac{i}{2}(N^*\partial_x+\partial_x N)$. 
\item The operator $\partial_x$ is unbounded, so we cannot directly apply Theorem \ref{theo:HP} to show the existence of a solution $U_t$, neither Theorem \ref{theo:AP} to show the convergence. Instead, we will break $\gu_{\tau, n}$ in two parts: one which is solution of a Hudson-Parthasarathy Equation with bounded coefficients, and one which is solution of a Hudson-Parthasarathy Equation with unbounded coefficients but which is very simple.
\item Hudson-Parthasarathy Equations with unbounded coefficients have been studied by Fagnola in \cite{Fagnola2006} and Fagnola and Wills in \cite{FagnolaWillsUnbounded03} in more complex cases and with more general method.
\end{enumerate}
\end{rem}

\begin{proof}
 We treat the convergence of $R_\tau$ and of $V_\tau$ separately. First, let us consider the isomorphism $R_\tau=D_\tau\otimes P_++D_{-\tau}\otimes P_-$ defined in Paragraph \ref{subsub:discreteOQBM}. We write $R_{\tau, n}=\calg_n R_\tau \calg_n^*$ the corresponding operator acting on the toy Fock space, and $r_{\tau, n}=(\calg_{\delta \Z}\otimes \calg_\tau)) R_{\tau,n} (\calg_{\delta \Z}\otimes \calg_\tau)^*$. Let us consider their product 
\[
z_{\tau, n}=r_{\tau,n}r_{\tau,n-1}\cdots r_{\tau,1}~.
\]

Note that $V_\tau$ is not acting on $\hh_z$ and $Z_\tau$ is not acting on $\hh_G$, so $\calg_n^* Z_\tau \calg_n$ commutes with $\calg_k^* V_\tau \calg_k$ for any $n>k$. Thus we have 
\[
\gu_{\tau, n}=z_{\tau, n} u_{\tau, n}~.
\]
We already know that $u_{\tau, \ent{t/\tau}}$ converges to some operator $U_t$ by Theorem \ref{theo:AP}. Let us consider the limit of the operator $z_{\tau, n}$.

{\bf The pointer process $Z_t$:}

\begin{prop}\label{prop:zt}
For any $t\in \R_+$ the operator $z_{\tau,\ent{t/\tau}}$ strongly converges to a unitary operator $Z_t$. The process $(Z_t)_{t\in \R_+}$ satisfies the following quantum SDE on the space $H^2(\R)\otimes_{alg}~ \eps(L^2(\R))$.
\begin{align}\label{eq:zt}
dZ_t&=\left(\frac{1}{2}\partial_x^2 dt-\partial_x(da^1_0(t)+da^0_1(t))\right) Z_t~.
\end{align}
In the probabilistic representation, $Z_t$ is explicit: for any function $f\in L^2(\R)$ and any random variable $A\in L^2(\calw, \mu)$ we have
\[
(Z_t fA)(x)=f(x-W_t)A~.
\]
\end{prop}

\begin{proof}
Note that $\calg_{\delta \Z}~ D_\tau~ \calg_{\delta \Z}^*=e^{-\partial_x}P_{\delta\Z}$, since $e^{-\delta\partial_x}$ is the operator of translation by $\delta$ on $L^2(\R)$. Moreover,
\[
\calg_\tau P_{\pm}(n) \calg_\tau^*=\frac{1}{2}\left(a^0_0(\tau,n)+a^1_1(\tau,n)\pm a^1_0(\tau,n)\pm a^0_1(n) \right)
\]
 so we have
\[
r_{\tau, n}=\frac{e^{-\delta \partial_x}+e^{\delta \partial_x}}{2}\left( a^0_0(\tau,n)+a^1_1(\tau,n)\right)P_{\delta\Z}+ \frac{e^{-\delta \partial_x}-e^{\delta \partial_x}}{2}\left(a^1_0(\tau,n)+ a^0_1(\tau,n) \right)P_{\delta \Z}~.
\]

We want to write $e^{-\delta\partial_x}\simeq I-\delta \partial_x+\frac{1}{2}\delta^2\partial_x^2$. Since $\partial_x$ is unbounded, it cannot be done directly. Let us consider the space $\dd_C\subset L^2(\R)$ of $C$-bandlimited functions for $C>0$ : Writing $\ff$ the Fourier transform, the space $\dd_C$ is defined as
\[
\dd_C=\set{f \in L^2(\R)~|~ \ff f \text{ is supported in }[-C,C]}~.
\] 

This space is stable by $\partial_x$ and $\bigcup_{C>0} \dd_C$ is dense in $L^2(\R)$. Restricted to $\dd_C$, the operator $\partial_x$ is bounded, so we can expand the exponential. However, the space $\dd_C$ is not stable by $P_\delta$, so we introduce
\[
\tilde{r}_{\tau, n}=\frac{e^{-\delta \partial_x}+e^{\delta \partial_x}}{2}\left( a^0_0(\tau,n)+a^1_1(\tau,n)\right)+ \frac{e^{-\delta \partial_x}-e^{\delta \partial_x}}{2}\left(a^1_0(\tau,n)+ a^0_1(\tau,n) \right)
\]
so that $r_{\tau, n}=\tilde{r}_{\tau,n}P_\delta$. We also write $\tilde{z}_{\tau,n}=\tilde{r}_{\tau,n}\tilde{r}_{\tau,n-1}\cdots\tilde{r}_{\tau, 1}$. Since $P_{\delta\Z}$ commutes with $\tilde{r}_{\tau, k}$ for all $k$, we have that $z_{\tau,n}=\tilde{z}_{\tau,n}P_{\delta\Z}$.
 The space $\dd_C$ is stable by $\tilde{r}_{\tau}$, and on this space, since $\partial_x$ is bounded we have
\[
\tilde{r}_{\tau, n}=\left(I+\frac{1}{2}\partial_x^2+O(\delta^3)\right)a^0_0(\tau,n)+O(\delta)a^1_1(\tau,n)+\left(-\delta \partial_x+O(\delta^2)\right) (a^1_0(\tau,n)+a^0_1(n))
\]
With $\delta=\sqrt{\tau}$, this sets us under the hypothesis of Theorem \ref{theo:AP}, with $K=0$ and $L=-\partial_x$.  Thus, $\tilde{z}_{\tau, \ent{t/\tau}}$ converges strongly (on $\dd_C$) to a unitary operator $Z^C_t$ which is solution of \eqref{eq:zt}. All the $Z_t^C$'s coincide on their common domain of definition, and they are unitary, so we can extend them to $H^2(\R)$ and $L^2(\R)$. They commute with $\partial_x$, so they are also unitary for the space $H^2(\R)$. Since the $\tilde{r}_{\tau,\ent{t/\tau}}$ are unitary and converge to $Z_t$ strongly on a dense subspace, they converge strongly on the full space. Moreover, $P_{\delta\Z}$ converges strongly to $I$, so $z_{\tau, \ent{t/\tau}}=\tilde{z}_{\tau,\ent{t/\tau}}P_{\delta\Z}$ also converge strongly to $Z_t$.

Finally, by the classical It\^o formula, for any $\cc^2$ function
\[
df(x-W_t)=f(x)-\int_0^t \partial_x f(x-W_s)dW_s+\frac{1}{2} \int_0^t \partial_x^2 f(x-W_s) ds~.
\] 
Thus, if we write $(\tilde{Z}_t fA)(x)=f(x-W_t)A$ for any $f\in L^2(\R)$, the processes $\tilde{Z}_t$ and $Z_t$ follow the same quantum SDE on $\cc^2$ functions. Since they have the same initial state $Z_0=I$, this implies that they are equal. 
\end{proof}

As a consequence of this proposition, the operators $\gu_{\tau,\ent{t/\tau}}$ converges to $\gu_t:=Z_t U_t$ and the It\^o product formula yields the stochastic Equation \eqref{eq:HPOQBM}.

\end{proof}

\begin{rem}
\begin{enumerate}
\item It is also possible to prove Theorem \ref{theo:HPOQBM} by using the Attal-Pautrat Theorem directly on $U_t$ restricted to $\hh_G\otimes_{alg}~ \dd_C \otimes_{alg}~ \Phi$ since $\hh_G\otimes_{alg}~\dd_C$ is stable by $H$ and $N$. However, we believe that the pointer unitary operator $Z_t$ has its own interest.
\item Note that $Z_t$ does not commute with $U_t$, we only have the commutation of $U_s$ and $Z_{t,s}:=Z_t Z_s^*$ for $s\leq t$. The formula $\gu_t=Z_t U_t$ is consistent with the construction of the discrete OQBM: we make the system evolve according to the unitary operator $U_t$, and we apply the operator $Z_t$ which implements the translation by $W_t$ to the position of the quantum particle.
\end{enumerate}
\end{rem}

\subsection{From the Hudson-Parthasarathy Equation to the Lindblad Equation}\label{subsub:Lindbladproof}

The family of operators $(\gu_t)_{0\leq t}$ and of states $\rho_{tot, t}=\gu_t (\rho_S\otimes \ket{\Omega}\bra{\Omega}) \gu_t^*$ consists in the most complete description of the OQBM. In this subsection, we show how the Lindbladian picture of the OQBM result from this unitary description.

Let us consider $\tilde{\Lambda}_{ \tau}$ as a quantum channel on $\hh_G\otimes \hh_z$ with the help of the isometry $\calg_{\delta\Z}~:~l^2(\delta\Z)\rightarrow L^2(\R)$; for $\rho\in \hh_G\otimes L^2(\R)$ we write
\[
\Lambda_{S, \tau}(\rho)=\calg_{\delta\Z}\tilde{\Lambda}_{ \tau}\left( \calg_{\delta\Z}^* \rho \calg_{\delta\Z}\right)\calg_{\delta\Z}^*
\]
 Now we are ready to study the convergence as $\tau\rightarrow 0$.
 
 \begin{prop}\label{prop:convergence_lindblad}
 For any initial state $\rho\in \gs(\hh_G\otimes\hh_z)$ and for any $t\geq 0$ the state $\Lambda_{S, \tau}^\ent{t/\tau}(\rho)$ converges in trace norm to
 \[
 \Lambda_S^t(\rho)=\tra_{\Phi}\left( \gu_t (\rho\otimes\ket{\Omega}\bra{\Omega})\gu_t^*\right)~.
 \]
 \end{prop}
 
 \begin{proof}
Because of the dilation of $\tilde{\Lambda}_\tau$ (Lemma \ref{lem:dilation}) we have
\[
\Lambda_{S, \tau}^n(\rho)=\tra_{\Phi} \left(\gu_{\tau,n} (\rho\otimes \ket{\Omega}\bra{\Omega})\gu_{\tau, n}^*\right)~.
\]
But by Theorem \ref{theo:HPOQBM} the state $\gu_{\tau,n} (\rho\otimes \ket{\Omega}\bra{\Omega})\gu_{\tau, n}^*$ converges in trace norm to $\gu_t (\rho\otimes\ket{\Omega}\bra{\Omega})\gu_t^*$, and the convergence is preserved by applying the partial trace.
 \end{proof}

 The semigroup $\Lambda_S^t$ is strongly continuous in $t$, but not continuous for the trace norm, so its generator is not defined on the whole space $\gs(\hh_G\otimes \hh_z)$ but only on the space of Sobolev states, as defined below.

\begin{defi}\label{defi:sobolev_state}
For any Hilbert space $\hh$ and any $k\in \N$ the set $W^k\gs(\hh, \hh_z)$ is the set of states $\rho$ on $\bb(\hh\otimes \hh_z)$ which admits a kernel $(x, y)\in \R^2 \mapsto K_\rho(x, y)\in \cals^1(\hh)$ which is in the Sobolev space $W^{k, 1}(\R^2,\cals^1(\hh))$. Equivalently, it is the space of states $\rho \in \gs(\hh\otimes \hh_z)$ such that for any $n\leq k$ the operator $[\rho, \abs{\partial_x}^n]$ is a bounded operator on $\hh\otimes W^{2,k}(\R)$.

The set  $W^k\gs(\hh, L^\infty(\R))$ is the set of states $\rho$ on $\bb(\hh)\otimes L^\infty(\R)$ which admits a kernel $x\in \R \mapsto K_\rho(x)\in \cals^1(\hh)$ which is in the Sobolev space $W^{k, 1}(\R, \cals^1(\hh))$. 
\end{defi}

We can now express the Lindblad Equation for $\Lambda_S^t$. 

\begin{theo}\label{theo:LindbladOQBM1}
For any initial state $\rho \in W^2\gs(\hh_G, \hh_z)$ the state $\rho_S(t)=\Lambda^t_S(\rho)$ is in $W^2\gs(\hh_G, \hh_z)$ for all $t>0$. Moreover, it satisfies the following equation:
\begin{align}\label{eq:LindbladOQBM1}
\frac{d}{dt}\rho_S(t)=\tilde{\call}(\rho_S(t))
\end{align}
where 
\[
\tilde{\call}(\rho)=-i[\tilde{H}, \rho]+L\rho L^*-\frac{1}{2}\left\{ L^*L, \rho\right\}
\]
with $L=N-\partial_x$ and $\tilde{H}=H-\frac{i}{2}\partial_x(N+ N^*)$. 

Writing $ K_t(x, y)$ the kernel of $\rho_\cc(t)$ this equation becomes
\begin{align}\label{eq:LindbladOQBMK}
\frac{d}{dt} K_t(x, y)=\call(K_t(x,y))+\frac{1}{2}\left(\frac{\partial}{\partial_x}+\frac{\partial}{\partial_y}\right)^2 K_t(x, y)-N\left(\frac{\partial}{\partial x}+\frac{\partial}{\partial y}\right)K_t(x,y)-\left(\frac{\partial}{\partial x}+\frac{\partial}{\partial y} \right)K_t(x,y) N^*~.
\end{align}
\end{theo}

Equation \eqref{eq:LindbladOQBM} can be formally obtained by writing $e^{-\delta \partial_x}\simeq I-\delta\partial_x +\frac{\tau}{2} \partial_x^2$, but let us prove it from the dilation of $\Lambda_S^t$.

\begin{proof}

{\bf The operator $\gu_t$ preserves the space $\cals_2(\hh_G\otimes \Phi, \hh_z)$ :}

 The operator $\gu_t$ commutes with $\partial_x$ (since $Z_t$ and $S_t$ both commute with $\partial_x$) so for any operator $\rho_{tot}\in W^2\gs(\hh_G\otimes \phi)$ we have $[\gu_t\rho \gu_t^*, \abs{\partial_x}^n]=\gu_t[\rho, \abs{\partial_x}^n] \gu_t^*$.

Thus, if $\rho\in W^2\gs(\hh_G, \hh_z)$ then $\gu_t (\rho\otimes \ket{\Omega}\bra{\Omega}) \gu_t^*\in W^2\gs(\hh_G\otimes\Phi, \hh_z)$ and so $\rho_S(t)\in W^2\gs(\hh_G, \hh_z)$.

To obtain the Lindblad Equation we use the Heisenberg representation: for any observable $A\in W^{2, 1}(\R^2, \bb(\hh_G))$ we have 
\[
\tr{\rho_{S,t}A}=\tr{\rho\bra{\Omega} \gu_t^* (A\otimes I_{\Phi}) \gu_t \ket{\Omega}}~.
\]
Using the It\^o formula applied to $\gu_t^* A \gu_t$ on the domain $\hh_G\otimes_{alg}~H^2(\R)\otimes_{alg}~\eps(L^2(\R))$, we obtain that 
\[
\gu_t^* A \gu_t=A+\int_0^t \gu_s^* \tilde{\call}^*(A) \gu_s ds+R_t
\]
where $R_t$ is a quantum stochastic integral with respect of terms of the form $da^i_j(t)$ with $(i, j)\neq (0,0)$, so that $\bra{\Omega}R_t\ket{\Omega}=0$. Thus
\begin{align*}
\tr{\rho_{S, t}A}-\tr{\rho A}&=\int_0^t \tr{ \gu_s(\rho\otimes \ket{\Omega}\bra{\Omega}) \gu_s^* ~\tilde{\call}^*(A)}~ds\\
&=\int_0^t \tr{\tilde{\call}(\rho_{S, t}(s)) A} ds
\end{align*}
which implies Equation \eqref{eq:LindbladOQBM1} by density of $W^{2, 1}(\R^2, \bb(\hh_G))$ in $\bb(hh_G\otimes \hh_z)$. The equation on the kernel $K_\rho$ is obtained by using the following formulas: if $\rho\in W^1\gs(\hh_G, \hh_z)$ then $\partial_x \rho$ and $\rho \partial_x$ are kernel operators, with
\begin{align}\label{eq:partialxrho}
K_{\rho \partial_x}(x, y)&=-\frac{\partial}{\partial y} K_\rho(x, y) \\
K_{\partial_x \rho}(x, y)&=\frac{\partial}{\partial_x} K_\rho(x,y)~.
\end{align}
These formulas are obtained through integrations by parts.
\end{proof}

Let us consider the restriction on the algebra $\mm=\bb(\hh_G)\otimes L^2(\R)$, in order to obtain the Lindblad Equation on \enquote{diagonal states} (Equation \eqref{eq:LindbladOQBM}). As noted in the introduction, we consider states restricted to $\mm$ rather than states whose density matrix is in $\mm$, since $\mm$ contains no non-trivial trace-class operators.

\begin{theo}\label{theo:LindbladOQBM2}
There exists a semigroup of super-operators $(\Lambda_{\mm}^t)_{0\leq t}$ on $\gs(\mm)$ such that for any state $\rho\in \gs(\hh_S)$ with restriction $\rho_{\mm}$ to $\mm$, the restriction to $\mm$ of the state $\rho_t=\Lambda_S^t(\rho)$ is $\Lambda_{\mm}^t(\rho_\mm)$~.

If a state $\rho_\mm$ admits a kernel $x\mapsto Q_\rho(x)$ which is in $W^{2, 1}(\R, \cals^1(\hh_G))$ then $\rho_{\mm,t}=\Lambda_\mm^t(\rho_\mm)$ also admits a kernel $Q_t\in W^{2, 1}(\R, \cals^1(\hh_G))$ and we have
\begin{align}\label{eq:LindbladOQBMQ}
\frac{d}{dt}Q_t(x)=\call(Q_t(x))+\frac{1}{2}\frac{\partial^2}{\partial x^2} Q_t(x)-\left(N\frac{\partial}{\partial x} Q_t(x)+\left(\frac{\partial}{\partial x} Q_t(x)\right)N^*\right)~.
\end{align}
\end{theo}

\begin{proof}

Note that $\gu_t^* \mm \gu_t \subset \mm\otimes \bb(\Phi)$ (because of the form of $Z_t$) so that for any $A\in \mm$ the expectancy $\tr{\Lambda_S^t(\rho) A}=\tr{(\rho\otimes\ket{\Omega}\bra{\Omega})\gu_t^* A \gu_t}$ only depends on the restriction of $\rho$ to $\mm$. This proves the existence of $\Lambda_\mm^t$. Equation  \eqref{eq:LindbladOQBMQ} is proved exactly the same way as Proposition \ref{theo:LindbladOQBM1}.
\end{proof}

\subsection{Hierarchy of the descriptions of the OQBM}

With the OQBM, we have many views on the same object, carrying more or less informations: 
\begin{enumerate}[label=\alph*)]
\item The state $\rho_{tot, t}=\gu_t (\rho_S\otimes \ket{\Omega}\bra{\Omega})\gu_t^*$ on $\hh_G\otimes \hh_z\otimes \Phi$ offers the most complete description.
\item The state $\rho_{tot, G, t}=U_t(\rho_G\otimes \ket{\Omega}\bra{\Omega})U_t^*=\tra_{\hh_z}(\rho_{tot, t})$ ignores the position of the particle, though its translation $W_t$ is still registered in $\Phi$.
\item The random state $\rrho_t$ with the random position $X_t$ ignores the quantum aspect of the position, but keeps tracks of the classical correlations between two different times.
\item The state $\rho_{S, t}=\tra_\Phi(\rho_{tot, t})=\Lambda_S^t(\rho_S)$ on $\bb(\hh_S)$ forgets about correlations between different times and the precise distribution of $\rrho_t$, but conserves a quantum view on the position.
\item The restriction of $\rho_{S, t}$ to $\mm=\bb(\hh_G)\otimes \cala_z$ with matrix density function $Q_t(x)=\E(\rrho_t |X_t=x)$: it forgets the correlations between different times and has only the classical information about the position. This is the smallest description where we have a closed equation for the evolution (Equation \eqref{eq:LindbladOQBMQ}) and which allows to compute the distribution of $X_t$.
\item The state $\rho_G=\tra_{\hh_z\otimes\Phi}(\rho_{tot, G, t})=\int_{x\in \R} Q_t(x) dx$ evolves according to the Lindbladian $\call$ and it completely ignores the position $X_t$.
\end{enumerate}

The descriptions a), c), d), e) are really dealing with the OQBM, while b) and f) are only considering the evolution on $\hh_G$.
They can be obtained one from another by partial traces, restriction and conditional expectancy according to the following hierarchy:
\vspace{0.5cm} 

 \begin{figure}[!h]\label{fig:hierarchy}
\begin{tikzpicture}
  \matrix (m) [matrix of math nodes,row sep=2em,column sep=2em,minimum width=2em]
  {
 {}& \rho_{tot, t} &{} \\     
     \rho_{tot, G, t} & \rrho_t & \rho_{S,t} \\   
    {} & Q_t(x) &{} \\
	{}&\rho_G &{} \\ 
     };
  \path[-stealth]
  	(m-1-2) edge node [left] {$\tra_{\hh_z}$} (m-2-1)
          	edge [dashed] node[right]{?}  (m-2-2)
          	edge node [right] {$\tra_{\Phi}$} (m-2-3)
  	(m-2-1) edge node [left] {$\tra_\Phi$} (m-4-2)
  	(m-2-3) edge node [right] {$|_{\mm}$} (m-3-2)
	(m-2-2) edge node [left] {$\E|X_t$} (m-3-2)
	(m-3-2) edge node [right] {$\int_\R$} (m-4-2);
 	\end{tikzpicture} 
 	\caption{Hierarchy between the descriptions of the Open Quantum Brownian Motion}
\end{figure}
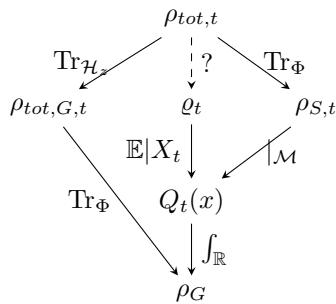

 The way we can obtain the process $(\rrho_t, X_t)$ directly from the unitary description is the subject of the second section of this article.



\section{Non-demolition measured evolution applied to the Open Quantum Brownian Motion}\label{sec:measure}

In the first section, we described the process $(\rrho_{\tau, n}, X_{\tau, n})_{n\in \N}$ as the result of a succession of unitary evolution by $L_\tau$ and measure of the position $X_{\tau, n}\in \delta \Z$. In continuous time this picture is harder to obtain, since the the measure and the evolution are happening at the same time. In this section we construct a general framework to deal with simultaneous measurement and evolution, using the crucial idea of non-demolition measurement introduced by Belavkin \cite{Belavkin94}.

\subsection{Evolution and measurement}

The evolution of a system after a measurement may be impossible to describe. Let us assume that the system evolves according to a unitary operator $U$ on $\hh_G\otimes L^2(\xx, \mu)$. 
We may measure the algebra $\cala=L^\infty(\xx, \nu)$ before \emph{or} after applying $U$, obtaining a random variable $X\in \xx$ and random states $\rrho_0=\rho_{G|\cala}$ and $\rrho_1=(U\rho U^*)_{S|\cala}$. However, it is not clear how to describe the measurement before \emph{and} after applying $U$. There may be two issues there:
\begin{enumerate}
\item The state $\rho_{GB|\cala}$ is well defined only if $\cala$ is discrete, else we only have the partial state $\rho_{G|\cala}$. Thus, we cannot define $U\rho_{GB|\cala}U^*$.
\item Even if $\cala$ is discrete, the measurement before applying $U$ modifies the state of the system, so $(U\rho_{GB|\cala} U^*)_{S|\cala}$ may not have the same distribution as $(U\rho U^*)_{S|\cala}$.
\end{enumerate}
The restriction to so called non-demolition evolutions allows to bypass these two issues in the general context of measurement under evolution. 

\begin{defi}\label{def:ha-nd}
Let $\hh_G$ and $\hh_B$ be two Hilbert spaces, let $I\subset \R$ be a set of times and $(U_t)_{t\in I}$ be a family of unitary operators on $\hh_G\otimes \hh_B$ with $U_0=I$ if $0\in I$ and let $(\cala_t)_{t\in I}$ be a family of commutative von Neumann algebras on $\hh_B$. Write $U_{t,s}=U_tU_s^*$ for any $s,t\in I$. We say that the process $(U_t, \cala_t)_{t\in I}$ is a $\hh_G$-non demolition evolution if for any $s\leq t \in I$ we have
\[
U_{t,s} \cala_s U_{t,s}^*\subset \set{I_G}\otimes \cala_t' 
\]
where $\cala_t'$ is the commutant of $\cala_t$. 
\end{defi}
In most cases the family of algebras will be increasing ($\cala_s\subset \cala_t$ for $s\leq t$) but we do not require it.

The non-demolition condition can be divided in two parts: the condition $U_{t,s} \cala_s U_{t,s}^*\subset  \bb(\hh_G)\otimes \cala_t'$ is here to ensure that the measure of $\cala_s$ does not disturb the measure of $\cala_t$ after evolution, while the condition $U_{t,s} \cala_s U_{t,s}^*\subset (I_G\otimes \bb(\hh_B))$ ensure that the random state at time $t$ is well defined. Let us describe more precisely how the random evolution can be defined.

Let us consider an $\hh_G$-non demolition evolution $(U_t, \cala_t)_{t\in I}$ and a state $\rho_0\in \gs(\hh_G\otimes \hh_B)$. We make the assumption that $I$ is upper bounded\footnote{this assumption is actually not necessary but it allows to use more concrete notations} by some $T\in I$. We fix some identifications $\cala_t\simeq L^\infty(\xx_t, \ff_t, \mu_t)$ implemented by some isometries $\calg_t : L^2(\xx_t, \ff_t, \mu_t)\rightarrow \hh_B$. We want to define a probability space $(\Omega, \P)$ with a stochastic process $(X_t, \rrho_t)_{t\in I}$ with $X_t \in \xx_t$ and $\rrho_t  \in \gs(\hh_G)$ obtained by simultaneously measuring $\cala_t$ at time $t$ and making evolve the system according to $U_t$. We construct it as follows.

\begin{itemize}
\item Let $\cala^U_t$ be the smallest von Neumann algebra containing all the algebras $U_{t,s} \cala_s U_{t,s}^*$ for $s\leq t$. It is commutative and contained in $I_G\otimes \bb(\hh_B)$ by the $\hh_G$-non demolition hypothesis. We fix an identification $\cala^U_t\simeq L^\infty(\xx^U_t,\ff^U_t \mu^U_t)$ implemented by an isometry $\calg^U_t: L^2(\xx^U, \ff^U_t, \mu^U_t)\rightarrow \hh_B$. 
\item For any $s\leq t$ we have $\cala_t\subset \cala^U_t$ so there exists a map $\phi_t:\xx^U_t\rightarrow \xx_t$ such that for any $f \in L^\infty(\xx_t, \ff_t, \mu_t)$ we have 
\[
\calg^U_t M_{f\circ \phi_t}(\calg^U_t)^*= \calg_t M_f \calg_t^*~.
\]
\item For $s\leq t$ we have $U_{t,s} \cala_s^U U_{t,s}^* \subset \cala_t^U $ so there are maps $\eta_{s,t}: \xx^U_t\rightarrow \xx^U_s$ such that for any $f\in L^\infty(\xx^U_s, \ff_s, \mu^U_s)$ we have
\[
\calg^U_t M_{f\circ \eta_{s,t}}(\calg^U_t)^*= U_{t,s} \calg^U_s M_f (\calg^U_s)^* U_{t,s}^*~.
\]
\item We take for our universe $\Omega$ the space $\xx^U_T$ with probability $\P=p^U_T d\mu^U_T$ induced by $U_T \rho U_T^*$ and the identification $\calg^U_T$. The random variable $X_t \in \xx_t$ is then defined as the measurable map $\phi_t \circ \eta_{t,T}$ from $\Omega=\xx^U_T$ to  $\xx_t$. 
\item The random variable $\rrho_t$ is defined as follows: we have a map $(U_t \rho U_t^*)_{G|\cala_t^U}$ on $\xx^U_t$, we compose it with $\eta_{t,T}$ to make it a map on $\xx^U_T$ : for $\omega\in\xx^U_T$,
\[
\rrho_t(\omega)=\big(U_t \rho U_t^*\big)_{G|\cala^U_t}(\eta_{t, T}(\omega))~.
\] 
\end{itemize}

\begin{rem} Note that the maps $\eta_{t,s}$ and $\phi_t$ are defined uniquely only up to a set of measure zero, as well as the random variable $\big(U_t \rho U_t^*\big)_{G|\cala^U_t}$. Thus, if $I$ is not countable there is not uniqueness in distribution of the process $(X_t, \rrho_t)_{t\in I}$, only uniqueness in finite-dimensional distributions. For example, when $\xx_t=\R$ for all $t$, the function $t\rightarrow X_t$ may be almost surely continuous, but this depends on the $\eta_{t,s}$ and $\phi_t$ which are chosen. 
\end{rem}

\begin{defi}
Any process $(X_t, \rrho_t)_{t\in I}$ obtained as above is called a measured evolution obtained from the $\hh_G$-non demolition evolution $(U_t, \cala_t)_{t\in I}$ and the state $\rho_0$. 
\end{defi}

This way of define the stochastic process should seem natural; a first motivation is that for all $t$ the variable $X_t$ has the same distribution as the result of the measure of $\cala_t$ in the state $U_t \rho U_t^*$, indeed for any function $f\in L^\infty(\xx_t, \ff_t, \mu_t)$ we have
\begin{align*}
\E(f)&=\E(f\circ \phi_t \circ \eta_{t,T})\\
&=\tr{U_T\rho U_T^* \calg^U_t M_{f\circ\phi_t \circ \eta_{t,T}} \calg_t^*}\\
&=\tr{U_t \rho U_t^*~U_{T,t}^*~ \calg^U_t M_{f\circ\phi_t \circ \eta_{t,T}} (\calg^U_t)^* U_{T,t}}\\
&=\tr{U_t\rho U_t^* ~\calg_t M_f \calg_t^*}~.
\end{align*}
However, this is only the distribution of $X_t$ at one time, and it does not justifies the joint distribution of the $X_t$'s for $t\in I$. We will use the indirect measurement (definition \ref{def:indirect_meas})to make a more complete and useful argument.

\begin{defi}\label{def:meas_evo}

For each $t$ let us fix an identification $\cala_t\simeq L^\infty(\xx_t, \ff_t, d\mu_t)$. We call an \emph{indirect measurement of $(\cala_t)_{t\in I}$ under the evolution $(U_t)_{t\in I}$} the following type of setup:
let $J=\set{t_0, \cdots, t_n} \subset I$ be a finite subset of $I$ and consider a family of pointer maps $(\psi_k)_{0\leq k\leq n}$ with $\psi_k : (\xx,\ff_{t_k})\times \yy_k\rightarrow \yy_k$ and a family of states $(\sigma_k)_{0\leq k \leq n}$ on $L^2(\yy_k, \nu_k)$ with corresponding probability density $p_k$ on $\yy_k$. Consider the pointer unitary operators $Z_k=Z_{\psi_k}$ as in Definition \ref{def:indirect_meas}. Let us perform successive indirect measurement: let $Y_{0}\in \yy_{t_0}$ be the result of the measurement of $L^\infty(\yy_{0})$ for the state  $Z_{0}\left(\big(U_{t_0}\rho U_{t_0}\big) \otimes \sigma_{0}\right) Z_{0}^*$, and $\rrho_{SB}(t_0)$ the state on $\hh_{S}\otimes \hh_B$ after the measurement; then, define $Y_{1}$ the result of the measurement of $L^\infty(\yy_{1})$ for the state $Z_{1} \left(\big(U_{t_1, t_0} \rrho_{SB}(t_0) U_{t_1, t_0}^*\big)\otimes \sigma_{1}\right) Z_{1}^*$, and define successively $Y_{2}, \cdots, Y_{n}$ the same way. We obtain a random process $(Y_k)_{0\leq k \leq n}$ on the space $\prod_{k=0}^n \yy_k$ and a family of random states
$\rrho^Y_{t_k}((Y_l)_{l\leq k})=\tra_{B} (\rho_{SB}(t_k))$~. 
\end{defi}

Note that we can perform this type of indirect measurement even if the property of $\hh_G$-non demolition is missing. The non-demolition property makes these indirect measurements to be consistent with the process described above, as follows. 

\begin{prop}[Consistency of the unraveling]\label{prop:consistency}

Consider any indirect measurement of $(\cala_t)_{t\in I}$ under the evolution $(U_t)_{t\in I}$ described as above. Assume that the $\hh_G$-non demolition property is satisfied.  Consider the random state $\rrho_t$ and the random variables $X_t \in \xx_t$ defined above on the universe $\xx_{tot}$. Add to this universe a family of random variables $(Y^0_t)_{t\in J}$ with probability distribution $p_k d\nu_k$, where $p_k$ is the probability density corresponding to the state $\sigma_k$ on $L^2(\yy_k, \nu_k)$. Assume that they are mutually independent and independent of $(X_t)_{t\in I}$ and define
\begin{align*}
\tilde{Y}_k&=\psi(X_{t_k}, Y^0_k) \\
\tilde{\rrho}^Y_k&=\E(\rrho_t((X_s)_{s\in I}) ~|~(\tilde{Y}_k)_{0\leq k \leq n})~.
\end{align*}
Then $(\tilde{Y}_k, \tilde{\rrho}^Y_{t_k})_{0\leq k \leq n}$ has the same distribution as the process $(Y_k, \rrho^Y_{t_k})_{0\leq k \leq n}$ defined by the indirect measurement.
\end{prop}

\begin{proof}
Let us write 
\[
W_k=Z_kU_{t_k, t_{k-1}} Z_{k-1} U_{t_{k_1}, t_{k-2}}\cdots Z_0 U_{t_0}~
\]
and let $\cala^Y_k=L^\infty(\prod_{l\leq k} \yy_k,~ \bigotimes_{l\leq k} \nu_l)$. 
Then 
\[
\rrho^Y_{t_k}=\left( W_k (\rho\otimes \sigma) W_k^*\right)_{S|\cala_k^Y}~.
\]
Moreover, for any function $f\in \cala^Y_k$ and operator $A\in \bb(\hh_S)$ we have
\begin{align*}
\E(\tr{\rrho^Y_{t_k}}f(Y_0, \cdots, Y_k)))&=\tr{W_k (\rho\otimes \sigma) W_k^* A\otimes M_f}~.
\end{align*}
Similarly, for all $l\leq k$ the random variable $\tilde{Y}_l$ corresponds to the measurable map $\psi_l(\phi_{t_l}\circ\eta_{t_l, t_k}(x_{t_k}),~y_l)$ on $\xx^U_{t_k}$ and so the random variable  $f(\tilde{Y_0}, \cdots, \tilde{Y_k})$ can be seen as the measurable map $g \in L^\infty(\xx^U_{t_k}\times \prod_{l\leq k} \yy_k)$ defined by
\[
g(x^U_{t_k}, y_0, \cdots, y_k)=f(\psi_0(\phi_{t_0}\circ\eta_{t_0, t_k}(x^U_{t_k}), y_0),\cdots, \psi_k(\phi_{t_k}(x^U_{t_k}), y_k))~.
\]
an by the construction of $\rrho_t$ and $\tilde{\rrho}^Y_t$ we have
\[
\E(\tr{\tilde{\rrho}^Y_{t_k}}f(\tilde{Y}_0, \cdots, \tilde{Y}_k)=\tr{(U_{t_k}\rho U_{t_k}^*\otimes \sigma) A\otimes \calg^U_t M_g(\calg^U_t)^*}
\]
Now, by the definition of $\eta_{t,t_k}$ and of $W_k$ we have 
\[
\calg^U_t M_g (\calg^U_t)^* = U_{t_k} W_k^* f W_k U_{t_k}^*
\]
by the definition of the $Z_k$ and $\phi_t, \eta_{t_k, t}$. Thus, 
\begin{align*}
\E(\tr{\rrho^Y_{t_k}}f(Y_0, \cdots, Y_k)))&=\tr{(U_{t_k}\rho U_{t_k}^*\otimes \sigma) A U_{t_k} W_k^* f W_k U_{t_k}^*}\\
&=\tr{(W_k\rho \otimes \sigma)U_{t_k}^* A U_{t_k} W_k^* f }~.
\end{align*}
Now, by $\hh_G$-non demolition, since $A$ is in the commutator of $I_G\otimes \bb(\hh_B)$ for any $l\leq k$  we have $U_{t_k, t_l}^* A U_{t_k, t_l}\in \cala_{t_l}'$ and in particular $U_{t_k, t_l}^* A U_{t_k, t_l}$ commutes with $Z_l^*$. Thus, we have
\begin{align*}
U_{t_k}^* A U_{t_k} W_k^*&=U_{t_0}^* (U_{t_0, t_k}^* A U_{t_k, t_0} ) Z_0 \cdots U_{t_{k-2}, t_{k-1}} Z_{k-1}^* U_{t_{k-1}, t_k} Z_k^* \\
&=  U_{t_0}^* Z_0^* U_{t_0, t_k}^* A U_{t_k, t_1} Z_1\cdots U_{t_{k-2}, t_{k-1}} Z_{k-1}^* U_{t_{k-1}, t_k} Z_k^*\\
\intertext{and with successive commutations we get}
U_{t_k}^* A U_{t_k} W_k^*&=W_k^* A ~.
\end{align*}
Thus we have 
\[
\E(\tr{\tilde{\rrho}^Y_{t_k}}f(\tilde{Y}_0, \cdots, \tilde{Y}_k))=\E(\tr{{\rrho}^Y_{t_k}}f(Y_0, \cdots, Y_k))~.
\]
This proves the equality in distribution.
\end{proof}

\subsubsection{The example of OQWs}

Open Quantum Random Walks are our first example of measured evolution. Let us consider any OQW $(B_e)_{e\in E}$ on a countable graph $(\vv, E)$. It consists in the succession of evolution by the quantum channel $\varphi(\rho)=\sum_{(x\rightarrow y)\in E} (B_{(x\rightarrow y)}\otimes \ket{y}\bra{x})\rho (B_{(x\rightarrow y)}^*\otimes \ket{x}\bra{y})$ and of measure of the algebra $\cala_\vv=l^\infty(\vv)$. As such, it does not need the formalism of measured evolution to be defined since $\cala_\vv$ is discrete, but it is a good demonstrator of measured evolution. 

Let us construct the auxiliary space $\hh_p=l^2(\vv)$. In the article \cite{OQWbirth} in which OQW where first defined, a unitary operator $U$ on $\hh_G\otimes l^2(\vv)\otimes \hh_p$  is constructed the following way: we fix a point $x_0\in \vv$. For any $x\in \vv$ we consider a unitary operator $V(x)$ such that for all $y\in \vv$ we have
\[
\bra{y}_{\hh_p} V(x) \ket{x_0}_{\hh_p}=\Un_{(x,y)\in \E} B_{(x\rightarrow y)}~.
\]
It exists because of the condition $\sum_{y \text{ with } (x\rightarrow y)\in E} B_e^* B_e=I$. Write $V(x)_{yz}=\bra{y}_{\hh_p}V(x)\ket{z}_{\hh_p}$. 
We put
\[
U=\sum_{x, y, z\in \vv} V(x)_{yz}\otimes \ket{y}\bra{x}\otimes \ket{x}\bra{z}~.
\]

Consider the Toy Fock space $T\Phi_\vv=\bigotimes_{n\in \N^*} \hh_p$ with respect to $\ket{x_0}$, and write $\ket{\Omega}=\bigotimes_{n\in\N^*} \ket{x_0}$. 
We consider the unitary operator $U(n, n-1)=\calg_n^* U \calg_n$ on $\hh_G\otimes l^2(\vv)\otimes T\Phi_\vv$ and define $U(n)=U(n,n-1)U(n-1,n-2)\cdots U(2,1)$. The system $(U(n), \cala_z)_{n\in \N}$ is $\hh_G$-non demolition, indeed $U(I_G\otimes \cala_\vv) U^* \subset I_G\otimes \cala_\vv\otimes \bb(\hh_p)$. More precisely, for any $f=\sum_{x\in \vv} f(x)\ket{x}\bra{x} \in \cala_z$ we have
\begin{align*}
U f U^*&=\sum_{x, y, y',z\in \vv} f(x) U(x)_{yz}U(x)_{y'z}\otimes \ket{y}\bra{y'}\otimes \ket{x}\bra{x}\\
&=\sum_{x, y\in \vv} f(x) I_G\otimes I\otimes \ket{x}\bra{x}~.
\end{align*}
Moreover, we have 
\[
\tra_{T\Phi_\vv}\left(U(n)(\rho\otimes \ket{\Omega}\bra{\Omega})U(n)^*  \right)=\varphi^n(\rho)~,
\]
 where $\varphi$ is the quantum channel defined by the OQW.
 By Proposition \ref{prop:consistency} this means that the OQW has the same distribution that the process $(\rrho_n, X_n)_{n\in \N^*}$ given by the measured evolution of $(U(n), \cala_\vv)_{n\in \N}$ with initial state $\rho\otimes\ket{\Omega}\bra{\Omega}$. Let us just make explicit the algebras $\cala^U_t$ and the maps $\phi_t$ and $\eta_{s,t}$ used in the definition of the measured evolution.

Writing $\cala_n=l^\infty(\vv^n)$ the algebra generated by the operators $\ket{x_1}\bra{x_1}\otimes \cdots\otimes \ket{x_n}\bra{x_n}\otimes I$ on $T\Phi_{\vv}$ we have
\[
\cala^U_n=\cala_z\otimes \cala_n=l^\infty(\vv\times \vv^n)~. 
\]
The operator $\phi_n: \vv\times \vv^n\rightarrow \vv$ is simply the projection on the first coordinate, and for $m< n$ the operator $\eta_{m,n}: \vv\times \vv^n\rightarrow \vv\times \vv^m$ is defined by
\[
\eta_{m, n}(x, x_1, \cdots, x_n)=(x_n, x_1, \cdots, x_m)~.
\]

\subsection{Application to the Open Quantum Brownian Motion}

With the measured evolution setup, we are able to obtain the process $(\rrho_t, X_t)_{0\leq t\leq T}$ satisfying the diffusive Belavkin Equation directly from the unitary operator $\gu_t$ and no more as the limit of a discrete-time repeated measurement setup. First, we just consider the system $(U_t, \cala_t)_{0\leq t \leq T}$ where $\cala_t=L^\infty(\calw([0, t]) \subset \bb(\Phi)$. Second, we apply this to the measured evolution of $(\gu_t, \cala_z)_{0\leq t \leq T}$ where $\cala_z=L^\infty(\R)\subset \bb(\hh_z)$. 

\subsubsection{Measured evolution for the Hudson-Parthasarathy process}

In this part we study the measured evolution $(U_t, \cala_t)_{0\leq t\leq T}$ on $\hh_G\otimes \Phi$. The setup is quite simple in this case, because $\cala^U_t=\cala_t$ and $\eta_{s, t}$ is just the map $(w_u)_{0\leq u\leq t}\rightarrow (w_u)_{0\leq u \leq s}$. This allows to study it in a less contrived way that the measured evolution described above, and the following result is well known in quantum filtering theory (see Theorem 7.1 and Corollary 7.2 in \cite{bouten_introduction_2007}).

\begin{prop}\label{prop:measevoBelavkin}
The system $(U_t, \cala_t)_{0\leq t\leq T}$ is $\hh_B$-non demolition. If $\hh_G$ is finite-dimensional it admits a measured evolution process $(\rrho_t, (W_s)_{s\leq t})_{0\leq t \leq T}$ corresponding to the initial state $\rho\otimes\ket{\Omega}\bra{\Omega}$ which satisfies the diffusive Belavkin Equation \eqref{eq:belavkin}.
\end{prop}

\begin{proof}
Note that for any $s$ the process of operators $(U_{s, t})_{s\leq t\leq T}$ satisfies the Hudson-Parthasarathy Equation \eqref{eq:HP} and $U_{s, s}=I$. Thus, $U_{s, t}$ does not act on $\Phi_{[0, s]}$, in particular for any $f\in \cala_s$ we have $U_{s, t} f U_{s, t}^*=f$. This proves the non-demolition. Since $U_{s, t} \cala_s U_{s, t}^*=\cala_s$ we have $\cala_t^U=\cala_t$, we can take $\phi_t$ the identity map on $\calw([0, t])$, and $\eta_{s, t}: \calw([0, t])\rightarrow \calw([0, s])$ is just the restriction to $[0, s]$. Thus, the state $\rrho_t$ satisfies
\[
\E_\P(\tr{\rrho_t A}f((W_u)_{u\leq t}))=\tr{U_t (\rho\otimes\ket{\Omega}\bra{\Omega}) U_t^* A\otimes f}
\]
for any observable $A\in \bb(\hh_G)$ and function $f\in \cala_t$. We study first the unnormalized state ${\ssig_t=u_{\calw([0, t])}\left(U_t (\rho\otimes\ket{\Omega}\bra{\Omega}) U_t^*\right)}$. It is defined as in Theorem \ref{theo:random_rho}, by
\[
\E_\mu(\tr{\ssig_t A}f((W_u)_{u\leq t}))=\tr{U_t (\rho\otimes\ket{\Omega}\bra{\Omega}) U_t^* A\otimes f}
\]
(where $\mu$ is the measure on $\calw([0, T])$ under which $(W_t)_{0\leq t\leq T}$ is the Wiener process). We compute the Equation for $\ssig_t$ using the It\^o formula. First, we use the Heisenberg representation:
\[
\tr{U_t (\rho\otimes\ket{\Omega}\bra{\Omega}) U_t^* A\otimes f}=\tr{\rho\bra{\Omega} U_t^*(A\otimes f) U_t \ket{\Omega}}~. 
\]
Let us write $f_s=\E_\mu(f|\ff_s)$ (where $\ff_s$ is the $\sigma$-algebra generated by $(W_u)_{u\leq s}$). It is a martingale; $f_s$ is bounded for all $s$ since $f$ is bounded, and by the predicable representation theory there exists an adapted process $(g_s)_{s\leq t}$ such that 
\[
df_s=g_s dW_s~
\]
or in terms of quantum SDE, $f_s=f_0+\int_0^s g_s(da^1_0(s)+da^0_1(s)$ on $\eps(L^2(\R))$. We apply the quantum It\^o formula two times to the product $U_s^*(A\otimes  f_s) U_s$; since we are interested in $\bra{\Omega}U_s^*(A\otimes  f_s) U_s\ket{\Omega}$ we can ignore the terms which are not in $dt$. We obtain
\begin{align*}
U_t^*(A\otimes  f_t) U_s&=A\otimes f_0+\int_0^t U_s^* \call^*(A)f_s U_s ds+\int_0^t U_s^*(N^* A+AN)U_s g_s ds+R_t~,
\end{align*}
where $R_t$ is a quantum It\^o integrals with only terms in $da^1_0(s)$ and $da^0_1(s)$. This implies that
\begin{align*}
\E_\mu(\tr{\ssig_t A} f((W_u)_{u\leq t})&=f_0\tr{\rho a}+\int_0^t \tr{U_s(\rho\otimes  \ket{\Omega}\bra{\Omega}) U_s \left(\call^*(A)f_s+(N^*A+AN)g_s \right)}ds\\
&=f_0\tr{\rho a}+\int_0^t\E_\mu\left(\tr{\call(\ssig_s)A} f_s+\tr{(N\ssig_s+\ssig_s N^*)A} g_s\right)~ds\\
&=f_0\tr{\rho a}+\E_\mu\left( \tr{\left(\int_0^t \call(\ssig_s)ds+\int_0^t (N\ssig_s+\ssig_s N^*)dW_s\right)A} f\right)
\end{align*}
the last equality being a consequence of the classical It\^o formula. This implies that, for $\hh_G$ of finite-dimension, 
\begin{align}\label{eq:ssig}
d\ssig_t=\call(\ssig_t)dt+(N\ssig_t+\ssig_t N^*)dW_t~.
\end{align}
It is now time to go back to $\rrho_t=\ssig_t/\tr{\ssig_t}$, and to compute the measure $\P$ with $d\P=\tr{\ssig_t} d\mu$. First, note that Equation \eqref{eq:ssig} has linear coefficients, so $\ssig_t$ is bounded in $L^2(\calw([0, T])$. Write $p_t=\tr{\ssig_t}$. Since $\tr{\call(A)}=0$ for any operator $A$, conditioned in $p_t\neq 0$ we have
\[
dp_t=\tr{N\ssig_t+\ssig_t N^*}dW_t=p_t\Tau(\rrho_t)dW_t~.
\]
Thus, $p_t$ is the exponential martingale 
\[
p_t=\exp{\int_0^t \Tau(\rrho_t)ds-\frac{1}{2}\int_0^T \Tau(\rrho_t)^2 ds }~.
\]
Note that $\E_\mu(p_T)=1$ by definition of $\ssig_t$, so it is indeed a martingale. By the Girsanov Theorem, under the distribution $p_Td\mu$ there exists a Wiener process $B_t$ defined by 
\begin{align}
B_0&=0\\
dB_t&=-\Tau(\rrho_t)dt+dW_t~.
\end{align}
This is the second line of Equation \eqref{eq:belavkin}. To compute the equation for $\rrho_t$, note that
\[
d\frac{1}{p_t}=d\exp{-\int_0^t \Tau(\rrho_t)ds+\frac{1}{2}\int_0^T \Tau(\rrho_t)^2 ds }=\frac{1}{p_t}\left(\Tau(\rrho_t)^2dt-\Tau(\rrho_t) dW_t\right)
\]
so with $\rrho_t=\ssig_t p_t^{-1}$ the It\^o formula yields the first line of Equation \eqref{eq:belavkin}.
\end{proof}

This derivation can be extended to more general Hudson-Parthasarathy Equations, and has also been studied in the case where the state on $\Phi$ is not $\ket{\Omega}\bra{\Omega}$ but a more complex, single-photon state, with a resulting non-markovian Belavkin Equation \cite{Gough12}.

\subsubsection{The measured evolution applied to the Open Quantum Brownian Motion}

The measured evolution of $(gu_t, \cala_z)_{0\leq t\leq T}$ is a little more subtle than the one of $(U_t, \cala_t)_{0\leq t\leq T}$, but it can be reduced to this last one by using the formula $\gu_t=Z_t U_t$. 

\begin{theo}\label{theo:measevoOQBM}
Assume that $\hh_G$ is finite-dimensional, and let us fix some $T>0$. Then the system$(U_t, \cala_z)_{t\in [0, T]}$ is $\hh_G$-non-demolition, and it admits a measured evolution $(\rrho_t, X_t)_{t\in [0, T]}$ which is almost surely continuous in time. It satisfies Equation \eqref{eq:belavkinOQBM}.
\end{theo}

\begin{proof}
For any $f\in \cala_z$ and any $s\leq t$ we have
\begin{align*}
\gu_{s, t} f \gu_{s, t}&=Z_{s,t} U_{s,t} f U_{s,t}^* Z_{s,t}^*\\
&=Z_{s,t} f Z_{s,t}^*
\end{align*}
which is the operator of multiplication by the function $\tilde{f}_{s, t}(x, (w_u)_{u\leq T})=f(x-w_s+w_t)$. Hence the system $(\gu_t, \cala_z)_{0\leq t\leq T}$ is $\hh_G$-non demolition, and we have $\cala_t^\gu=\cala_z\otimes\cala_t=L^\infty(\R\times\calw([0, t]), Leb\otimes\mu)$. We choose the map $\phi_t: \R\times\calw([0, t])\rightarrow \R$ as the projection on the first coordinate, and for $s\leq t$ we take the map $\eta_{s, t}: \R\times\calw([0, t])\rightarrow \R\times\calw([0, s])$ defined by
\[
\eta_{s, t}(x, (w_u)_{0\leq u\leq t})=(x-w_t+w_s,~ (w_u)_{0\leq u\leq s})~.
\]
Let $(\rrho_t, X_t)_{0\leq t\leq T}$ be the random measured process corresponding to these maps. Write $h=\left(\gu_t (\rho\otimes \ket{\Omega}\bra{\Omega})\gu_t^*\right)_{G|\cala^\gu_t}$ (it is a random variable on $\R\otimes\calw([0, t])$), then $\rrho_t$ is the random variable on $\R\otimes\calw([0, T])$ defined by
\[
\rrho_t(x, (w_u)_{0\leq u\leq T})=h(x-w_T+w_t, (w_u)_{0\leq u\leq T})~. 
\]
For any $x\in \R$ consider the random variable on $\calw([0, T])$ obtained by conditioning $\rrho_t$ to $X_0=x$. This random variable is $\nu_t(x)=\rrho_t(x+w_t,(w_u)_{0\leq u\leq T})$. By definition of $Z_t$ it is actually equal to 
\[
\left(U_t (\nu_0(x)\otimes\ket{\Omega}\bra{\Omega}) U_t^*\right)_{G|\cala_t}~.
\]
Thus, $(\nu_t(x), W_t)_{0\leq t\leq T}$ is the random evolution corresponding to the measured evolution of $(U_t, \cala_t)_{0\leq t\leq T}$ with initial state $\nu_0(x)$, and by the definition of $\eta_t$ we have $X_t=X_T-W_T+W_t=X_0+W_T$ so Proposition \ref{prop:measevoBelavkin} yields Equation \eqref{eq:belavkinOQBM}.
\end{proof}

\subsection{Towards general convergence theorems for measured evolution}

The convergence of $\rho_{t,\tau}=\Lambda_{\tau}^\ent{t/\tau}(\rho)$ to $\rho_t=\Lambda_S^t(\rho)$ was obtained directly from the strong convergence of $\gu_{\tau, t}$ to $\gu_t$. On the contrary, the convergence in distribution of $(\rrho_{\tau, t}, X_{\tau, t})_{O\leq t\leq T}$ to $(\rrho_t, X_t)_{0\leq t\leq T}$ was shown as a consequence of Pellegrini's Theorem \ref{theo:HPOQBM}, which was proved by classical probabilistic methods without any reference to the operators $U_t$ on the Fock space and on the measured evolution. 

A natural question is: can we prove the convergence in distribution of a family of processes $(\rrho_{\tau, t}, X_{\tau, t})_{0\leq t\leq T}$ coming from a measured evolution $(U_{t, \tau}, \cala_\tau)_{0\leq t\leq T}$  just from the strong convergence of $U_{\tau, t}$ and of the algebras $\cala_\tau$ to some operator $U_t$ and some algebra $\cala$ ? 

This question turns out to be rather difficult, since the algebra $\cala^{U_\tau}_{\tau,t}$ also depends in $(U_{\tau, t})_{0\leq t\leq T}$. In what follows we present some results in this direction. 
\newline
A first result can be obtained when there is no evolution and we are only considering one measurement. 

\begin{prop}\label{prop:random_rho_cv}
Let $\rho_1, \rho_2 \in \gs\left(\hh_G\otimes L^2(\xx ,\mu)\right)$ be two states and let $\cala=L^\infty(\xx, \mu)$. For $i=1,2$ define the random variables $\rrho_i=(\rho_i)_{G|\cala}$ on $(\xx, \P_i)$ where $d\P_i=p_id\mu$ are defined as in Theorem \ref{theo:random_rho}. Then 
\begin{align}
\norm{p_1-p_2}_{L^1(\xx, \mu)}\leq \norm{\rho_1-\rho_2}_{\cals^1(\hh_G\otimes L^2(\xx)}\\
E_{\P_1}(\norm{\rrho_1-\rrho_2}_{\cals^1(\hh_G})\leq 2\norm{\rho_1-\rho_2}_{\cals^1(\hh_G\otimes L^2(\xx)}~.
\end{align}
\end{prop}
\begin{proof}
Write $h_i=u_{\xx}(\rho_i)$ the unnormalized states corresponding to $\rho_i$. Then $p_i(x)=\tr{h_i(x)}$ for $\mu$-almost every $x\in \xx$ so 
\[
\norm{p_1-p_2}_{L^1(\xx, \mu}\leq \int_{\xx} \tr{\abs{h_1(x)-h_2(x)}}d\mu(x)\leq \norm{\rho_1-\rho_2}_{\cals^1}
\]
the last inequality being a consequence of the fact that $h_i$ is the restriction to $\bb(\hh_G)\otimes \cala$ of the state $\rho_i$. Thus, 
\begin{align*}
E_{\P_1}(\norm{\rrho_1-\rrho_2}_{\cals^1(\hh_G})&= \int_{\xx} \tr{\abs{\rrho_1(x)-\rrho_2(x)}}p_1d\mu(x)\\
&\leq \int{\xx} \tr{\abs{p_1(x)\rrho_1(x)-p_2(x)\rrho_2(x)}}d\mu(x)+\int_{\xx}\tr{\abs{(p_1(x)-p_2(x))\rrho_2(x)}}d\mu(x)\\
&\leq 2\norm{\rho_1-\rho_2}_{\cals^1}~.
\end{align*}
\end{proof}

As a consequence we have the following: 

\begin{cor}\label{cor:random_rho_cv}
Let $(\rho_n)_{n\in \N}$ be a sequence of states on $\hh_G\otimes L^2(\xx, \mu)$ converging in $\cals^1(\hh_G\otimes \hh_B)$ to some state $\rho$. Consider the sequence of random variables $\rrho_n=\rho_{G|\cala}$ defined as in Theorem \ref{theo:random_rho}. Then $\rrho_n$ converges to $\rrho$ in distribution and in $L^1(\xx, \cals^1(\hh_S), p_\rho d\mu)$.
\end{cor}

Note that it would make no sense to ask that $\rrho_n$ converge to $\rrho$ in probability or almost surely since they are attached to different probability measures on $\xx$. The convergence in $L^1(\xx, \cals^1(\hh_S), p_\rho d\mu)$ is already a little strange from a probabilistic point of view though it is mathematically meaningful: the random state $\tilde{\rho_n}$ is in $L^1(\xx, \cals^1(\hh_S), p_\rho d\mu)$ since it is bounded in $\cals^1(\hh_S)$ and $p_\rho d\mu$ is a probability measure.
\newline

 In the case of measured evolutions, we were only able to obtain the following partial result, in which the convergence of the process $(X_{tau,t_n})_{t_n\in I_n}$ is obtained, but not the convergence of the random state.

\begin{prop}\label{prop:meas_evo_cv}
 Let $\xx=\R^d$ with Borelian algebra $\ff$ and a Radon measure $\mu$. For each $n\in \N$ let $\ff_n$ be a coarse sub-$\sigma$-algebra of $\ff$. Assume $\ff_n\subset \ff_{n+1}$ for each $n$ and write $\xx_n=\R^d/\ff_n$. identified with subsets of $\R^d$ such that $\xx_n\subset \xx_{n+1}\subset \xx$. We fix some time set $I=[0,T]$ upper-bounded by some $T\in \R$ and some finite set $I_n\subset I$ with $I_n\subset I_{n+1}$.

  Consider some Hilbert spaces $\hh_G$ and $\hh_C$ and write $\hh_B=L^2(\xx,\ff, \mu)\otimes \hh_C$. 
Consider $\cala=L^\infty(\xx,\ff, \mu)$ and let $(U_t, \cala)_{t\in I}$ be an $\hh_G$-non demolition measured evolution and $\rho$ a state on $\hh_G\otimes \hh_B$. We write $(X_t)_{t\in I} \in \xx^I$ and $(\rrho_t)_{t\in I}$ the random variables obtained by measuring $\cala$ under the evolution.

 For each $n\in \N$ fix a closed subspace $\hh_{n,C}\subset \hh_C$ with $\hh_{n,C}\subset \hh_{n+1,C}$. Write $\hh_n=L^2(\xx, \ff_n, Leb)\otimes \hh_{n, C}$ and let $P_n$ the orthogonal projection on $\hh_n$. Note that $P_n$ commutes with every elements of $\cala$, we define $\cala_n=P_n \cala$ and $\xx_n=\R^d/\ff_n$.
 Consider a process of unitary operators $(U_{n,t})_{t\in I_n}$ on $\hh_G\otimes\hh_n$ (that we may see as partial isometries on $\hh_G\otimes \hh_B$), and a state $\rho^n$ on $\hh_G\otimes \hh_n$ (that we may see as a state on $\hh_G\otimes \hh_B$). Assume that $(U_{n,t}, \cala_n)_{t\in I_n}$ is $\hh_G$-non demolition for all $t$. Define the process $(X_{n,t})_{t\in I_n}$ with values in $\xx_n$ and $(\rrho_{n,t})_{t\in I_n}$ obtained by the measured evolution of $\cala_n$ under  the evolution $U_{n,t}$ with initial state $\rho^n$. We still write $t\in I\rightarrow X_{n, t}$ the extension of $t\in I_n\rightarrow X_{n t}$ to $I$ by linear interpolation, and the same for $\rrho_{n, t}$. 
 
We make the following assumptions: 

\begin{assumption}\label{as:ant}
Writing $I_n=\{t_{1, n}, \cdots, t_{k_n, n}\}$ (in increasing order) we assume that 
\[
l_n=max\left\{t_{i+1,n}-t_{i, n}|1\leq i \leq k_n\right\}
\]
converges to $0$ as $n\rightarrow \infty$. 
\end{assumption}

\begin{assumption}\label{as:density_f}
 For any $x\in \R^d$ write 
\[
C_{\ff_n}(x)=\bigcap_{A\in \ff_n,~x\in A} A~.
\]
Then we assume that
\[
\lim_{n\rightarrow \infty} sup_{x\in \R^d} ~diam(C_{\ff_n}(x))=0~.
\]
\end{assumption}

\begin{assumption}\label{as:tightness}
The sequence of processes $(X_{n,t})_{t\in I}$ is tight for the topology of the uniform convergence on the set of continuous functions on $I$, and $(X_t)_{t\in I}$ is almost surely continuous.
\end{assumption}

\begin{assumption}\label{as:cvu}
The sequence of projections $(P_n)_{n\in \N}$ strongly converges to the identity and the state $\rho^n$ converges to $\rho$ in $\bb^1$ as $n \rightarrow 0$ and for all sequence $(t_n)_{n\in \N}$ with $t_n\in I_n$ converging to some $t\in I$ the operator $U_{n, t_n}$ strongly converge to $U_t$ on $\hh_G\otimes \hh_B$.
\end{assumption}

 Then $(X_{n,t})_{t\in I}$ converges in distribution (in the topology of uniform convergence) to $(X_t)_{t\in I}$.
\end{prop}

In the case of the OQBM, we choose a sequence $\tau_n$ such that $\delta_n/\delta_{n+1}\in \N$. We have $\xx=\R$ and $\xx_n=\delta_n\Z$, the algebra $\ff_n$ being generated by the sets $[\delta k, \delta (k+1)~)$ and we take $\hh_C=\Phi$ and $\hh_{C, n}=T_{\tau_n}\Phi$. Upon proving the tightness assumption \ref{as:tightness}, this theorem together with Theorem \ref{theo:measevoOQBM} provides an alternative proof of the convergence of $(X_{n,\ent{t/\tau}})_{t\in [0,T]}$ to a process solution of \eqref{eq:belavkinOQBM}. However, it is very incomplete since we do not prove the convergence of $\rrho_{n, \ent{t/\tau}}$.

Note that Assumption \ref{as:tightness} depends on the maps $\eta_{s,t}$ and $\phi_t$ chosen in the construction of the process, which are only defined up to a set of measure zero.

\begin{proof}
We separate the dependency on $I_{n}$ and $\cala_n$ on the one hand and on $U_{n, t}$ on the other hand. For $k\leq n$ and any $t\in n I_k$ we write 
\[
X_{k, n, t}=C_{\ff_k}(X_{ n, t})
\]
and we consider the $\sigma$-algebra $\ff_{k, n}$ generated by $(X_{k, n, t})_{t\in I_k}$ and define
\[
\rrho_{k, n, t}=\E(\rrho_{n, t}|\ff_{k, n})~.
\]
Then $(\rrho_{k, n, t},X_{k, n, t})_{t\in I_k}$ is a measured evolution corresponding to the system $(U_{n, t}, \cala_k)_{t\in I_k}$. We also write
\[
X_{k, \infty, t}=C_{\ff_k}(X_t)~
\]
and $\ff_{k, \infty}$ the $\sigma$-algebra generated by $(X_{k, \infty,t})_{t\in I_k}$, and $\rrho_{k, \infty, t}=\E(\rrho_t|\ff_{k, \infty})$, so that $(\rrho_{k,\infty, t}, X_{k, \infty, t})_{t\in I_k}$ is a measured evolution corresponding to the system $(U_t, \cala_k)_{t\in I_k}$. We extend all these functions to $I$ by linear interpolation. 

We prove the convergence in distribution of $(X_{n,t})_{t\in I}$. Let $f$ be a bounded Lipschitz function on the space $\dd$ of continuous functions from $[0, T]$ to $ \xx$. We want to show that $\E(f((X_{n,t})_{t\in [0, T]}))$ converges to $\E(f(( X_t)_{t\in [0, T]}))$ as $n\rightarrow \infty$.

We fix $\eps>0$. For any $k$ sufficiently large, we have $diam(C_{\ff_k}(x))\leq \eps$ for all $x\in \xx$. By the tightness assumption, there is $C>0$ such that for any $n$ sufficiently large,  with probability higher than $1-\eps$  we have $\norm{X_{n,t}-X_{n,s}}\leq \eps $ if $\abs{t-s}\leq C$. Since $d_k\rightarrow 0$ as $n\rightarrow \infty$ this implies that for all $n$ and $k$ large enough we have $\norm{X_{k,n,t}-X_{n,t}}\leq 2\eps $ for all $t$  with probability higher than $1-\eps$. Writing $M=\max \abs{f}$ and $L$ the Lipschitz constant for $f$, this means that there is $K\in \N$ such that for any $n,k\geq K$,
\begin{align}\label{eq:proofcv1}
\abs{\E(f((X_{k, n, t})_{t\in I}))-\E(f((X_{n,t})_{t\in I}))}\leq \eps M+ 2\eps L~.
\end{align}
The crucial point is that this bound is uniform in $n$. The same reasoning shows that for any $k$ large enough we have
\begin{align}\label{eq:proofcv2}
\abs{\E(f(( X_{k, \infty, t})_{t\in I}))-\E(f(( X_{t})_{t\in I}))}\leq \eps M +2\eps L~.
\end{align}
Thus we can fix some $k$ such that the two above quantities are less than $\eps$, and compare $( X_{k, \infty, t})_{t\in I_k}$ and $(X_{k, n, t})_{t\in I_k}$. They are both measurements of discrete algebras on a discrete set of times, so we can actually describe them as indirect measurement, as follows.

We write $I_{k, n}=\set{t_1, \cdots, t_m}$ with $t_0<t_1<\cdots < t_m$. For $1\leq l\leq m$ consider some copies $\yy_l$ of $\xx_{k}$ and write $\hh_l=L^2(\yy_l, \nu)$ with $\nu$ the counting measure. We fix $a\in \xx_{k}$ and define the state $\sigma_l=\ket{a}\bra{a}$ on $\hh_l$. We consider a pointer map $\psi:\xx_{k}\times \xx_{k}\rightarrow \xx_{k}$ such that $\psi(x,a)=x$ for all $x\in \xx_{k}$, and we define the pointer unitaries $Z_l$ on $\hh_B\otimes \hh_l$ as in Definition \ref{def:indirect_meas}. Write $\hh_Y=\bigotimes_{0\leq l \leq m} \hh_l$ and $\sigma=\sigma_0\otimes\cdots\otimes\sigma_m$ and:
\begin{align*}
W&=Z_m U_{t_m, t_{m-1}}Z_{m-1} U_{t_{m-1}, t_{m-2}}Z_{m-2}\cdots Z_0 U_{t_0}~. \\
W_{{n}}&=Z_m U_{n,t_m, t_{m-1}}Z_{m-1} U_{n,t_{m-1}, t_{m-2}}Z_{m-2}\cdots Z_0 U_{n,t_0}~.
\end{align*}

Consider the states $\rho_{W}=W(\rho\otimes \sigma )W^*$ and $\rho_{W_{n}}=W_{n}(\rho^{n}\otimes \sigma)W_{n}^*$.  Now, write $\cala_Y=L^\infty(\prod_{0\leq l\leq m}\yy_l, \nu)$, then the result of the measurement of $\cala_Y$ in the state $\rho_W$ is a process $(Y_t)_{t\in I_{k}}$. By Proposition \ref{prop:consistency} it has the same distribution as $(X_{k,\infty,t})_{t\in I_{k}}$. Likewise, the result of the measurement of $\cala_Y$ in the state $\rho_{W_{n}}$ is a process $(Y_{n,t})_{t\in I_{k}}$ with same distribution as $(X_{k,n, t})_{t\in I_{k, n_1}}$. Now by Assumption \ref{as:cvu} the operator $W_{n}$ converges strongly to $W$ as $n_2\rightarrow n$ and $\rho^{k}$ converges to $\rho$ in $\bb^1$ so $\rho_{W_{k}}$ converges to $\rho_{W}$ in $\bb^1$, so by Proposition \ref{prop:random_rho_cv} the process $(Y_{n, t})_{t\in I_{k}}$ converges in distribution to the process $(Y_{t})_{t\in I_{k}}$.

This implies that for $n$ large enough, 
\[
\abs{\E(f(( X_{k, \infty, t})_{t\in I}))-\E(f(( X_{k,n,t})_{t\in I}))}\leq \eps~.
\]
But $k$ was fixed large enough so that $\abs{\E(f((X_{k, n, t})_{t\in I}))-\E(f((X_{n,t})_{t\in I}))}\leq \eps$  this implies that 
\[
\abs{\E(f(( X_{n, t})_{t\in I}))-\E(f(( X_{t})_{t\in I}))}\leq 3\eps~
\]
thus proving the convergence in distribution of $(X_{n, t})_{t\in I}$. 
\end{proof}

The key point is the inequality \eqref{eq:proofcv1} which is uniform in $n$. Such a uniform estimate could not be obtained for $\rrho_{k,n,t}$. Indeed, even if the $\sigma$-algebra $\ff_{k, n}$ is very close to the full $\sigma$-algebra for $k$ large enough, this does not implies that $\rrho_{k,n,t}=\E(\rrho_{n,t}|\ff_{k,n})$ is close to $\rrho_{n,t}$ for $k$ large enough uniformly in $n$. 

The hypothesis that $\cala_n$ is coarse and $I_n$ is finite is actually not necessary. To go without it, we may use coarse subalgebras $\cala_{k, n}$ of $\cala_n$ and finite subsets $I_{n,k}\subset I_n$ and look at the measured evolutions of $(U_{n,t}, \cala_{k,n})_{t\in I_{n,k}}$.

\subsection{Open questions and prospects}

Three questions are left open in Theorem \ref{prop:meas_evo_cv}:
\begin{enumerate}
\item Are the assumptions sufficient to ensure the convergence in distribution of $(\rrho_{n,t})_{t\in I}$ ?
\item On what condition does an $\hh_G$-non-demolition system $(U_t, \cala_t)_{0\leq t \leq T}$ admit a measured evolution process $(\rrho_t, X_t)_{0\leq t\leq T}$ which is almost surely continuous in time? It is the case for $U_t$ defined by the Hudson-Parthasarathy Equation and $\cala_t=L^\infty(\calw([0, t]), \mu)$, but it is not the case when $\cala_t$ is the algebra generated by the $a^1_1(s)$ for $s\leq t$ (the measured evolution has jumps in this case, see for example \cite{pellegrini_jumps_10}).
\item Considering a family of $\hh_G$-non demolition systems $(U_{\tau, t}, \cala_\tau)_{t\in I_\tau}$ with measured evolutions $(\rrho_{t, \tau}, X_{t, \tau})_{0\leq t\leq T}$. Is there any condition on the unitaries and algebras to ensure the tightness of the family of processes in the space of continuous functions? 
\end{enumerate}

Some questions concern the OQBM more specifically. 

\begin{enumerate}[resume]
\item In the trajectories of the Open Quantum Brownian Motion, there is no back-action of the position $X_t$ on the state $\rrho_t$, which satisfies a closed equation. This framework is insufficient in the context of quantum control, where we would want $X_t$ to represent some control function which depends on the history of the trajectory. What if $N$ and $H$ depends on the position $X_t$ ? We may expect that under some regularity condition on the functions $x\mapsto N(x)$ and $x\mapsto H(x)$ (for example, Schwartz functions), there exists an \emph{inhomogeneous OQBM}, whose unitary operator $\gu_t$ is solution of the equation
\[
d\gu_t=\Big((-iM_H-\frac{1}{2}M_{N^*N}+\frac{1}{2}\partial_x^2-\partial_x M_N)dt+(M_N-\partial_x)da^0_1(t)+(-M_N^*-\partial_x)da^1_0(t) \Big) \gu_t~
\]
where $M_N$ is the operator on $\hh_G\otimes \hh_z=L^2(\R, \hh_G)$ defined by $M_Nf(x)=N(x)f(x)$. This idea was raised in the original article on the OQBM, \cite{OQBM}. Formally, everything works the same way as the homogeneous OQBM, the equation for the measured evolution being expected to be of the form
\[
\left\{
\begin{array}{ll}
d\rrho_t&=\call_{X_t}(\rrho_t)dt+\left(N(X_t)\rrho_t+\rrho_t N(X_t)^*-\rrho_t\Tau_{X_t}(\rho_t)\right)dB_t\\
dX_t&=\Tau_{X_t}(\rrho_t)dt+dB_t~.
\end{array}
\right.
\]
However, proving the existence of $\gu_t$ is far more complex than for the homogeneous OQBM, since the operators $\partial_x$ and $M_N$ are no more commuting, and the space of bandlimited functions $\dd_C$ is no more preserved. The existenc of solutions of Hudson-Parthasarathy Equations with unbounded coefficients have been studied in \cite{Fagnola2006} and \cite{FagnolaWillsUnbounded03}, but the convergence of discretisations in the toy Fock space in the spirit of Attal and Pautrat has never been studied for unbounded operators.
\item The generalization of the homogeneous OQBM to higher dimensions is straightforward. Going further, we may study an inhomogeneous OQBM on a manifold. With an Einstein manifold for example, this may provide a semiclassical model for a relativistic quantum particle, in the spirit of the relativistic Brownian motion \cite{angst2016}, \cite{franchiLeJan}.
\end{enumerate}

\bibliography{/home/andreys/Documents/in_progress/bibli_quantum.bib}

\end{document}